\documentclass[12pt]{article}
\usepackage{algorithmic}
\usepackage{amsmath,amssymb,amsfonts,amsthm,verbatim}

\usepackage[numbers]{natbib}
\usepackage{graphicx}
\usepackage{epstopdf}
\usepackage{subcaption}
\usepackage{url}
\def\tB{\widetilde{B}}

\usepackage{float}
\def\bX{{\bf X}}
\newcommand{\real}{\ensuremath{\mathbb{R}}}

\floatstyle{ruled}
\newfloat{algorithm}{tbhp}{loa}[section]
\floatname{algorithm}{Algorithm}
\newcommand{\half}{\mbox{$\frac12$}}
\newtheorem{cor}{Corollary}
\newtheorem{thm}{Theorem}
\newtheorem{remark}{Remark}
\newtheorem{lem}{Lemma}

\newcommand{\B}{\boldsymbol}
\newcommand{\M}{\mathbf}
\newcommand{\sgn}{\operatorname{sgn}}

\newcommand{\rnk}{\mathrm{rank}}

\DeclareMathOperator*{\mini}{minimize}

\def\tU{{\tilde U}}
\def\tV{{\tilde V}}
\def\tD{{\tilde D}}
\newcommand{\Soft}[1]{{\cal S}_{#1}}
\newcommand{\BALD}{\begin{aligned}}
\newcommand{\EALD}{\end{aligned}}
\newcommand{\BALDS}{\begin{aligned*}}
\newcommand{\EALDS}{\end{aligned*}}
\newcommand{\BCAS}{\begin{cases}}
\newcommand{\ECAS}{\end{cases}}
\newcommand{\BEAS}{\begin{eqnarray*}}
\newcommand{\EEAS}{\end{eqnarray*}}
\newcommand{\BEQ}{\begin{equation}}
\newcommand{\EEQ}{\end{equation}}
\newcommand{\BIT}{\begin{itemize}}
\newcommand{\EIT}{\end{itemize}}
\newcommand{\BMAT}{\begin{bmatrix}}
\newcommand{\EMAT}{\end{bmatrix}}
\newcommand{\BNUM}{\begin{enumerate}}
\newcommand{\ENUM}{\end{enumerate}}

\newcommand{\BA}{\begin{array}}
\newcommand{\EA}{\end{array}}

\DeclareMathOperator*{\argmin}{\arg\min}
\newcommand{\diag}{\mathop{\mathbf{diag}}}

\newcommand{\minimize}{\operatornamewithlimits{minimize}}

\DeclareMathOperator{\rank}{rank}

\DeclareMathOperator{\trace}{tr}

\newcommand{\norm}[1]{\left\| #1 \right\|}

\begin{document}
\title{Matrix Completion and Low-Rank SVD via Fast Alternating Least Squares}
\author{Trevor Hastie \and Rahul Mazumder \and Jason D. Lee \and Reza Zadeh}
\date{Statistics Department and ICME\\Stanford University\\[8pt]\today}
\maketitle
\begin{abstract}
  The matrix-completion problem has attracted a lot of attention,
  largely as a result of the celebrated Netflix competition.  Two
  popular approaches for solving the problem are
  nuclear-norm-regularized matrix approximation
  \citep{candes-2009,mazumder09:_spect_regul_algor_for_learn}, and
  maximum-margin matrix factorization \citep{natis-05}. These two
  procedures are in some cases solving equivalent problems, but with
  quite different algorithms. In this article we bring the two
  approaches together, leading to an efficient algorithm for large
  matrix factorization and completion that outperforms both of these.
 We develop a software package
  \texttt{softImpute} in R for implementing our approaches, and a
  distributed version for very large matrices using the {\tt Spark}
  cluster programming environment
\end{abstract}
\section{Introduction}\label{sec:introduction}
We have an $m\times n$ matrix $X$ with observed entries
indexed by the set $\Omega$; i.e. $\Omega=\{(i,j): X_{ij}\mbox{ is
  observed}\}.$ Following \citet{candes-2009} we define the
projection $P_\Omega(X)$ to be the $m\times n$ matrix with the
observed elements of $X$ preserved, and the missing entries replaced
with $0$. Likewise $P_\Omega^\perp$ projects onto the complement of
the set $\Omega$.

Inspired by \citet{candes-2009}, \citet{mazumder09:_spect_regul_algor_for_learn} posed the following convex-optimization problem for completing $X$:
\begin{equation}
  \label{eq:1}
  \mini_M\;\; H(M):= \;\; \mbox{$\frac12$}\|P_\Omega(X-M)\|_F^2 +\lambda \|M\|_*,
\end{equation}
where the {\em nuclear norm} $\|M\|_*$ is the sum of the singular values of $M$ (a convex relaxation of the rank).
They developed a simple iterative algorithm for solving (\ref{eq:1}), with the following two steps iterated till convergence:
\begin{enumerate}
\item Replace the missing entries in $X$  with the corresponding entries from the current estimate $\widehat M$:
  \begin{equation}
    \label{eq:xhatstep}
    \widehat X \leftarrow P_\Omega(X) +P_\Omega^\perp(\widehat M);
  \end{equation}
\item Update $\widehat M$ by computing the soft-thresholded SVD of $\widehat X$:
  \begin{eqnarray}
    \widehat X&=&UDV^T    \label{eq:softSVD1}\\
    \widehat M&\leftarrow&U\Soft{\lambda}(D)V^T, \label{eq:softSVD2}
  \end{eqnarray}
  where the soft-thresholding operator $\Soft{\lambda}$ operates
  element-wise on the diagonal matrix $D$, and replaces $D_{ii}$ with
  $(D_{ii}-\lambda)_+$. With large $\lambda$ many of the diagonal
  elements will be set to zero, leading to a low-rank solution for
  (\ref{eq:1}).  
\end{enumerate}

For large matrices, step (\ref{eq:softSVD1}) could be
  a problematic bottleneck, since we need to compute the SVD of the filled
  matrix $\widehat X$.  In fact, for the Netflix problem $(m,n)\approx
  (400K,20K)$, which requires storage of $8 \times 10^9$ floating-point numbers
  (32Gb in single precision), which in itself could pose a
  problem. However, since only about 1\% of the entries are observed (for the Netflix dataset),
  sparse-matrix representations can be used.

\citet{mazumder09:_spect_regul_algor_for_learn} use two tricks to avoid these computational nightmares:
\begin{enumerate}
\item Anticipating a low-rank solution, they compute a reduced-rank SVD in step~(\ref{eq:softSVD1}); if the smallest of the computed singular values is less than $\lambda$, then this gives the desired solution. A reduced-rank SVD can be computed with alternating subspace methods, which can exploit warms start (which would be available here).
\item They rewrite $\widehat X$ in (\ref{eq:xhatstep}) as
  \begin{equation}
    \label{eq:splr}
    \widehat X = \left[P_\Omega(X)-P_\Omega(\widehat M)\right] +\widehat M;
  \end{equation}
The first piece is as sparse as $X$, and hence inexpensive to store and compute. The second piece is low rank, and also inexpensive to store. Furthermore, the alternating subspace methods mentioned in step (1) require left and right multiplications of $\widehat X$ by \emph{skinny} matrices, which can exploit this special structure.
\end{enumerate}

This \texttt{softImpute} algorithm works very well, and although an SVD needs to be computed each time step~(\ref{eq:softSVD1}) is  evaluated,
this step can use the previous solution as a warm start. As one gets closer to the solution, the warm starts tend to be better, and so the final iterations tend to be faster.

\citet{mazumder09:_spect_regul_algor_for_learn} also considered a path of such solutions, with decreasing values of $\lambda$. As $\lambda$ decreases, the rank of the solutions tend to increase, and at each $\lambda_\ell$, the iterative algorithms can use the solution  ${\widehat X}_{\lambda_{\ell-1}}$ (with $\lambda_{\ell-1}>\lambda_{\ell} $) as  warm starts, padded with some additional dimensions.

\citet{natis-05} consider a different approach. They impose a rank constraint, and consider the problem
\begin{equation}
  \label{eq:mmmf}
  \mini_{A,B} F(A,B):=\frac{1}{2}\|P_\Omega(X-AB^T)\|_F^2 +\frac{\lambda}{2}(\|A\|_F^2+\|B\|_F^2),
\end{equation}
where $A$ is $m\times r$ and $B$ is $n\times r$. This so-called
maximum-margin matrix factorization (\texttt{MMMF}) criterion is not convex in
$A$ and $B$, but it is bi-convex. They and others use alternating minimization algorithms (\texttt{ALS}) to solve (\ref{eq:mmmf}). Consider $A$ fixed, and we wish to solve (\ref{eq:mmmf}) for $B$. It is easy to see that this problem decouples into $n$ separate ridge regressions, with each column $X_j$ of $X$ as a response, and the $r$-columns of $A$ as
predictors. Since some of the elements of $X_j$ are missing,
and hence ignored, the corresponding rows of $A$ are deleted for the $j$th
regression. So these are really \emph{separate} ridge regressions, in
that the regression matrices are all different (even though they all
derive from $A$). By symmetry, with $B$ fixed, solving for $A$ amounts to $m$ separate ridge regressions.

There is a remarkable fact that ties the solutions to  (\ref{eq:mmmf}) and (\ref{eq:1}) \citep[for example]{mazumder09:_spect_regul_algor_for_learn}. If the solution to 
(\ref{eq:1}) has rank $q\leq r$, then it provides a solution to (\ref{eq:mmmf}). That solution is 
\begin{equation}
  \begin{array}{rcl}
\label{eq:equiv}
 \widehat A&=&U_{r}\Soft{\lambda}(D_{r})^{\half}\\
 \widehat B&=&V_{r}\Soft{\lambda}(D_{r})^{\half},
\end{array}
\end{equation}
where $U_r$, for example, represents the sub-matrix formed by the first $r$ columns of $U$, and likewise $D_r$ is the top $r\times r$
diagonal block of $D$. Note that for any solution to (\ref{eq:mmmf}), multiplying $\widehat A$ and $\widehat B$ on the right by an orthonormal $r\times r$ matrix $R$ would be an equivalent solution. Likewise, any solution to (\ref{eq:mmmf}) with rank $r\geq q$ gives a solution to (\ref{eq:1}).

In this paper we propose a new algorithm that profitably draws on ideas used both in \texttt{softImpute} and \texttt{MMMF}.
Consider the two steps (\ref{eq:softSVD1}) and (\ref{eq:softSVD2}). We can alternatively solve the optimization problem
\begin{equation}
  \label{eq:als1}
  \mini_{A,B} \frac{1}{2}\|\widehat X-AB^T\|_F^2+\frac{\lambda}{2}(\|A\|_F^2+\|B\|_B^2),
\end{equation}
and as long as we use enough columns on $A$ and $B$, we will have $\widehat M=\widehat A{\widehat B}^T$.
There are several important advantages to this approach:
\begin{enumerate}
\item Since $\widehat X$ is fully observed, the (ridge) regression operator is the same for each column, and so is computed just once. This reduces the computation of an update of $A$ or $B$ over \texttt{ALS} by a factor of $r$.
\item By orthogonalizing the $r$-column $A$ or $B$ at each iteration, the regressions are simply matrix multiplies, very similar to those used in the alternating subspace algorithms for computing the SVD. 
\item The ridging amounts to shrinking the higher-order components more than the lower-order components, and this tends to offer a convergence  advantage over the  previous approach (compute the SVD, then soft-threshold).
\item Just like before, these operations can make use of the {\em sparse plus low-rank} property of $\widehat X$.
\end{enumerate}

As an important additional modification,  we replace $\widehat X$ at each step using the most recently computed $\widehat A$ or $\widehat B$.
All combined, this hybrid algorithm tends to be faster than either approach on their own; see the simulation results in Section~\ref{sec:timing}

For the remainder of the paper, we present this \texttt{softImpute-ALS} algorithm in more detail, and show that it convergences to the solution to (\ref{eq:1}) for $r$ sufficiently large. We demonstrate its superior performance on simulated and real examples, including the Netflix data.
We briefly highlight two publicly available software implementations, and describe a simple approach to centering and scaling of both the rows and columns of the (incomplete) matrix.

\section{Rank-restricted Soft SVD}
\label{sec:rrssvd}
In this section we consider a complete matrix $X$, and develop a new algorithm for finding a rank-restricted SVD. In the next section we will adapt this approach to the matrix-completion problem.
We first give two theorems that are likely known to experts; the proofs are very short, so we provide them here for convenience.


\begin{thm}
\label{thm1}
  Let $X_{m\times n}$ be a matrix (fully observed), and let $0<r\leq \min(m,n)$. Consider the optimization problem
  \begin{equation}
    \label{eq:opt1}
    \mini_{Z:\; \mbox{\rm rank}(Z)\leq r}F_\lambda(Z):=\half||X-Z||_F^2 +\lambda||Z||_*.
  \end{equation}
A solution is given by 
\begin{eqnarray}
  \label{eq:soln}
\widehat Z&=& U_r \Soft{\lambda}(D_r) V_r^T,
\end{eqnarray}
where the rank-$r$ SVD of $X$ is $U_rD_rV_r^T$ and $ \Soft{\lambda}(D_r) = \diag[(\sigma_1-\lambda)_+,\ldots,(\sigma_r-\lambda)_+]$.
\end{thm}
\begin{proof}
We will show that, for any $Z$ the following inequality holds:
\begin{equation}\label{thm1-ineq-1}
 F_\lambda(Z) \geq f_\lambda(\B\sigma(Z)):= \half||\B\sigma(X)-\B\sigma(Z) ||_2^2 +\lambda \sum_{i} \sigma_{i}(Z),
 \end{equation}
where, $f_\lambda(\B\sigma(Z))$ is a function of the singular  values of  $Z$
and $\B\sigma(X)$ denotes the vector of singular values of $X$, such that $\sigma_{i}(X) \geq \sigma_{i+1}(X)$ for all $i=1, \ldots, \min \{ m,n \}$.

To show inequality~\eqref{thm1-ineq-1} it suffices to show that:
$$ \half||X-Z||_F^2   \geq \half||\B\sigma(X)-\B\sigma(Z) ||_2^2, $$
which follows as an immediate consequence of the by the well known Von-Neumann or Wielandt-Hoffman trace inequality~(\cite{stewart-sun:1990}):
$$\langle X,  Z \rangle \leq \sum_{i=1}^{\min\{m,n\}} \sigma_{i}(X) \sigma_i (Y).$$
 Observing that
 $$ \mbox{\rm rank}(Z)\leq r  \iff  \|\B\sigma(Z) \|_0 \leq r,$$
  we have established:
  \begin{equation}\label{eq:vec-mat-ineq1}
\begin{aligned}
   \inf_{Z:\; \mbox{\rm rank}(Z)\leq r}& \;\;  \left  (  \half||X-Z||_F^2 +\lambda||Z||_*\right )    \\ 
 \geq    \inf_{\B\sigma(Z) :  \|\B\sigma(Z) \|_0 \leq r}&  \left  ( \half||\B{\sigma}(X)-\B\sigma(Z) ||_2^2 +\lambda \sum_{i} \sigma_{i}(Z) \right )
\end{aligned}
  \end{equation}
  Observe that the optimization problem in the right hand side of~\eqref{eq:vec-mat-ineq1} is a separable vector optimization problem.
  We claim that the optimum solutions of the two problems appearing in~\eqref{eq:vec-mat-ineq1} are in fact equal. 
 To see this, let 
$$ \widehat{\B\sigma(Z)}  = \argmin_{\B\sigma(Z) :  \|\B\sigma(Z) \|_0 \leq r}\;\;  \left  ( \half||\B{\sigma}(X)-\B\sigma(Z) ||_2^2 +\lambda \sum_{i} \sigma_{i}(Z) \right ).$$
If the SVD of $X$ is given by $UDV^T$, then the choice $\widehat{Z}=U \mathrm{diag}(\widehat{\B\sigma(Z)}) V'$ satisfies 
$$\rnk(\widehat{Z}) \leq r \;\; \text{and} \;\;\; F_\lambda(\widehat{Z}) = f_\lambda(\widehat{\B\sigma(Z)})$$
This shows that:
  \begin{equation}\label{eq:vec-mat-eq1}
\begin{aligned}
   \inf_{Z:\; \mbox{\rm rank}(Z)\leq r} & \left  (  \half||X-Z||_F^2 +\lambda||Z||_*\right ) \\
  = \inf_{\B\sigma(Z) :  \|\B\sigma(Z) \|_0 \leq r} &  \left  ( \half||\B{\sigma}(X)-\B\sigma(Z) ||_2^2 +\lambda \sum_{i} \sigma_{i}(Z) \right )
\end{aligned}
  \end{equation}
and thus concludes the proof of the theorem.

\end{proof}
This generalizes a similar result where there is no rank restriction, and the problem is convex in $Z$.
For $r<\min(m,n)$,  \eqref{eq:opt1} is not convex in $Z$, but the solution can be characterized in terms of the SVD of $X$.

The second theorem relates this problem to the corresponding matrix factorization problem

\begin{thm}
\label{thm2}
  Let $X_{m\times n}$ be a matrix (fully observed), and let $0<r\leq \min(m,n)$. Consider the optimization problem
  \begin{equation}
    \label{eq:opt2}
  \min_{A_{m\times r},\;B_{n\times r}} \frac{1}{2}\norm{X-AB^T}_F^2 +\frac{\lambda}{2}\left(\norm{A}_F^2+\norm{B}_F^2\right)
  \end{equation}
A solution is given by $\widehat A=U_r \Soft{\lambda}(D_r)^{\frac12}$ and $\widehat B=V_r \Soft{\lambda}(D_r)^{\frac12}$, and all solutions satisfy 
$\widehat A\widehat B^T=\widehat Z$, where, $\widehat Z$ is as given in~\eqref{eq:soln}.
\end{thm}
We make use of the following lemma from \cite{natis-05,mazumder09:_spect_regul_algor_for_learn}, which we give without proof:
\begin{lem}
  \label{lem:lem1}
\[\norm{Z}_*=\min_{A,B: Z=AB^T} \frac{1}{2}\left( \norm{A}_F ^2 +\norm{B}_F ^2 \right)\]
\end{lem}
\begin{proof}
Using Lemma~\ref{lem:lem1}, we have that
\begin{align*}
  &\min_{A_{m\times r},\;B_{n\times r}} \frac{1}{2}\norm{X-AB^T}_F^2 +\frac{\lambda}{2}\norm{A}_F^2+\frac{\lambda}{2} \norm{B}_F^2\\
  &=\min_{Z: \rank(Z)\le r}  \frac{1}{2}\norm{X-Z}_F^2 +\lambda \norm{Z}_*
\end{align*}
The conclusions follow from Theorem \ref{thm1}.
\end{proof}
Note, in both theorems the solution might have rank less than $r$.

Inspired by the alternating subspace iteration algorithm~\cite{golub2012matrix} for the reduced-rank SVD, we now present Algorithm~\ref{alg:rrssvd}, an alternating ridge-regression algorithm for finding the solution to (\ref{eq:opt1}).
\begin{algorithm}[p]
  \caption{Rank-Restricted Soft SVD}
\label{alg:rrssvd}
\begin{enumerate}
\item Initialize $A=UD$ where $U_{m\times r}$ is a randomly chosen matrix with orthonormal columns and $D=I_r$, the $r\times r$ identity matrix.
\item Given $A$, solve for $B$:
  \begin{equation}
    \label{eq:BgivenA}
    \minimize_B||X-AB^T||_F^T+\lambda||B||_F^2.
  \end{equation}
This is a multiresponse ridge regression, with solution
\begin{equation}
  \label{eq:Bsol}
  \tilde{B}^T=(D^2+\lambda I)^{-1}DU^TX.
\end{equation}
This is simply matrix multiplication followed by coordinate-wise shrinkage.
\item Compute the SVD of $\tilde{B}D = \tilde{V}\tilde{D}^2R^T$, and let $V\leftarrow \tilde{V}$, $D\leftarrow \tilde{D}$, and $B=VD$.
\item Given $B$, solve for $A$:
  \begin{equation}
    \label{eq:BgivenA}
    \minimize_A||X-AB^T||_F^T+\lambda||A||_F^2.
  \end{equation}
This is also a multiresponse ridge regression, with solution
\begin{equation}
  \label{eq:Bsol}
  \tilde{A}=XVD(D^2+\lambda I)^{-1}.
\end{equation}
Again  matrix multiplication followed by coordinate-wise shrinkage.
\item Compute the SVD of $\tilde{A}D = \tilde{U}\tilde{D}^2R^T$, and let $U\leftarrow \tilde{U}$, $D\leftarrow \tilde{D}$, and $A=UD$.
\item Repeat steps (2)--(5) until convergence of $AB^T$.
\item Compute $M=XV$, and then it's SVD: $M=U D_\sigma R^T$. Then output $U$, $V\leftarrow VR$ and $\Soft{\lambda}(D_{\sigma})=\mbox{diag}[(\sigma_1-\lambda)_+,\ldots,(\sigma_r-\lambda)_+]$.
\end{enumerate}
\end{algorithm}
\paragraph{Remarks}

\begin{enumerate}
\item At each step the algorithm keeps the current solution in ``SVD'' form, by representing $A$ and $B$ in terms of orthogonal matrices. The computational effort needed to do this is exactly that required to perform each ridge regression, and once done makes the subsequent ridge regression trivial. 
\item The proof of convergence of this algorithm is essentially the same as that for an alternating subspace algorithm~\cite{golub2012matrix} (without shrinkage). 


\item In principle step (7) is not necessary,  but in practice it cleans up the rank nicely.
\item This algorithm lends itself naturally to distributed computing for very large matrices $X$; $X$ can be chunked into smaller blocks, and the left and right matrix multiplies can be chunked accordingly. See Section~\ref{sec:distr-impl}.
\item There are many ways to check for convergence. Suppose we have a pair of iterates $(U,D^2,V)$ (old) and $(\tU,\tD^2,\tV)$ (new), then  the relative change in Frobenius norm is given by
  \begin{eqnarray}
  \nabla F&=&\frac{||UD^2V^T-\tU\tD^2\tV^T||_F^2}{||UD^2V^T||_F^2}\nonumber\\
&=&\frac{\trace(D^4)+\trace(\tD^4)-2\trace(D^2U^T\tU\tD^2\tV^TV)}{\trace(D^4)},   \label{eq:conv}
  \end{eqnarray}
which is not expensive to compute.
\item If $X$ is sparse, then the left and right matrix multiplies can be achieved efficiently by using sparse matrix methods.
\item Likewise, if $X$ is sparse, but has been column and/or row centered (see Section~\ref{centerscale}), it can be represented in ``sparse plus low rank'' form; once again left and right multiplication of such matrices can be done efficiently.
\end{enumerate}

An interesting feature of this algorithm is that a reduced rank SVD of $X$ is available from the solution, with the rank determined by the particular value of $\lambda$ used. The singular values would have to be corrected by adding $\lambda$ to each. There is empirical evidence that this is faster than without shrinkage, with accuracy biased more toward the larger singular values.


\section{softImpute-ALS Algorithm}
\label{eMMMF}
Now we return to the case where $X$ has missing values, and the non-missing entries are indexed by the set $\Omega$. We present algorithm~\ref{alg:rrsimpute} (\texttt{softImpute-ALS}) for solving (\ref{eq:mmmf}):
$$  \mini_{A,B}\|P_\Omega(X-AB^T)\|_F^2 +\lambda(\|A\|_F^2+\|B\|_F^2).
$$
where $A_{m\times r}$ and $B_{n\times r}$ are each of rank at most $r\leq\min(m,n)$. 
\begin{algorithm}
  \caption{Rank-Restricted Efficient Maximum-Margin Matrix Factorization: \texttt{softImpute-ALS} }
\label{alg:rrsimpute}
\begin{enumerate}
\item Initialize $A=UD$ where $U_{m\times r}$ is a randomly chosen matrix with orthonormal columns and $D=I_r$, the $r\times r$ identity matrix, and $B=VD$ with $V=0$. Alternatively, any prior solution  $A=UD$ and $B=VD$ could be used as a warm start.
\item Given $A=UD$ and $B=VD$, approximately solve
  \begin{equation}
    \label{eq:BgivenA2}
    \minimize_{\tB}  \frac{1}{2} \|P_\Omega(X-A\tB^T)\|_F^T+\frac{\lambda}{2}\|\tB\|_F^2
  \end{equation}
to update $B$.
We achieve that with the following steps:
  \begin{enumerate}
  \item Let $X^*= \left(P_\Omega(X)-P_\Omega(AB^T)\right)+AB^T$, stored as \emph{sparse plus low-rank}.
  \item Solve \begin{equation}
  \label{eq:oneside}
  \minimize_{\tB} \frac{1}{2}\|X^*-A\tB^T\|_F^2 +\frac{\lambda}{2} \|\tB\|_F^2,
\end{equation}
with solution
\begin{eqnarray}
  \tB^T &=& (D^2+\lambda I)^{-1}DU^TX^*\nonumber\\
&=&(D^2+\lambda I)^{-1}DU^T \left(P_\Omega(X)-P_\Omega(AB^T)\right)  \label{eq:solq2}\\
&&\quad + (D^2+\lambda I)^{-1}D^2 B^T.\label{eq:second}
\end{eqnarray}
\item Use this solution for $\tB$ and update $V$ and $D$:
  \begin{enumerate}
  \item Compute $\tilde B D=\tU \tD ^2\tV^T$;
  \item $V\leftarrow \tU$, and $D\leftarrow \tD$.
  \end{enumerate}
  \end{enumerate}

\item Given $B=VD$, solve for $A$. By symmetry, this is equivalent to step~2, with $X^T$ replacing $X$, and $B$ and $A$ interchanged.

\item Repeat steps (2)--(3) until convergence.
\item Compute $M=X^*V$, and then it's SVD: $M=U D_\sigma R^T$. Then output $U$, $V\leftarrow VR$ and $D_{\sigma,\lambda}=\mbox{diag}[(\sigma_1-\lambda)_+,\ldots,(\sigma_r-\lambda)_+]$.
\end{enumerate}
\end{algorithm}

The algorithm exploits the decomposition
\begin{equation}
  \label{eq:splr}
  P_\Omega(X-AB^T)=P_\Omega(X)+P_\Omega^\perp(AB^T)-AB^T.
\end{equation}
Suppose we have current estimates for $A$ and $B$, and we wish to compute the new $\widetilde B$.
We will replace the first occurrence of $AB^T$ in the right-hand side of  (\ref{eq:splr}) with the current estimates, leading to a {\em filled in} $X^*=P_\Omega(X)+P_\Omega^\perp(AB^T)$, and then solve for $\widetilde B$ in 
$$  \mini_{\widetilde B}\|X^*-A\widetilde B\|_F^2 +\lambda\|\widetilde B\|_F^2.
$$

  Using the same notation, we can write (importantly) 
\begin{equation}
  \label{eq:splr2}
  X^*=P_\Omega(X)+P_\Omega^\perp(AB^T)= \left(P_\Omega(X)-P_\Omega(AB^T)\right)+AB^T;
\end{equation}
This is the efficient {\em sparse + low-rank} representation for high-dimensional problems; efficient to store and also efficient for left and right multiplication.

\paragraph{Remarks}

\begin{enumerate}
\item This algorithm is a slight modification of Algorithm~\ref{alg:rrssvd}, where in step 2(a) we use the latest imputed matrix $X^*$ rather than $X$.
\item The computations in step 2(b) are particularly efficient. In (\ref{eq:solq2}) we use the current version of $A$ and $B$ to predict at the observed entries $\Omega$, and then perform a multiplication of a sparse matrix on the left by a skinny matrix, followed by rescaling of the rows. In (\ref{eq:second}) we simply rescale the rows of the previous version for $B^T$.
\item\label{item:3} After each update, we maintain the integrity of the current solution. By Lemma~\ref{lem:lem1} we know that the solution to
\begin{equation}
  \label{eq:minnorm}
  \minimize_{A,\;B\;:AB^T=\tilde A\tilde B^T}(\|A\|_F^2+\|B\|_F^2)
\end{equation}
is given by the SVD of $\tilde A\tilde B^T=UD^2V^T$, with $A=UD$ and $B=VD$. Our iterates maintain this each time $A$ or $B$ changes in step~2(c), with no additional significant computational cost.
\item The final step is as in Algorithm~\ref{alg:rrssvd}. We know the solution should have the form of a soft-thresholded SVD. The alternating ridge regression might not exactly reveal the rank of the solution. This final step tends to clean this up, by revealing exact zeros after the soft-thresholding.
\item In Section~\ref{sec:theoretical-results} we discuss (the lack of) optimality guarantees of fixed points of Algorithm~\ref{alg:rrsimpute} (in terms of criterion (\ref{eq:1}). We note that the output of \texttt{softImpute-ALS} can easily be piped into \texttt{softImpute} as a warm start. This typically exits almost immediately in our \texttt{R} package \texttt{softImpute}.
\end{enumerate}


\section{Theoretical Results}
\label{sec:theoretical-results}

 
 In this section we investigate the theoretical properties of the \texttt{softImpute-ALS} algorithm
 in the context of problems~\eqref{eq:mmmf} and~\eqref{eq:1}.
  
  We show that the~\texttt{softImpute-ALS} algorithm converges to a first order stationary point for 
  problem~\eqref{eq:mmmf} at a rate of $O(1/K),$ where $K$ denotes the number of iterations of the algorithm.
We also discuss the role played by $\lambda$ in the convergence rates.
We establish the limiting properties of the estimates produced by the \texttt{softImpute-ALS} algorithm: properties of the 
limit points of the sequence $(A_k,B_{k})$ in terms of problems~\eqref{eq:mmmf} and~\eqref{eq:1}. 
We show that for any $r$ in problem~\eqref{eq:mmmf}
the sequence produced by the 
\texttt{softImpute-ALS} algorithm leads to a decreasing sequence of objective values for the convex problem~\eqref{eq:1}.
A fixed point of the \texttt{softImpute-ALS} problem need not correspond to the minimum of the convex problem~\eqref{eq:1}.
We derive simple necessary and sufficient conditions that must be satisfied for a stationary point of the algorithm to be a 
minimum for the problem~\eqref{eq:1}---the conditions can be verified by a simple structured low-rank SVD computation.

\medskip
 
We begin the section with a formal description of the updates produced by the \texttt{softImpute-ALS} algorithm in terms of 
the original objective function~\eqref{eq:mmmf} and its majorizers~\eqref{defn-QA} and~\eqref{defn-QB}.
Theorem~\ref{thm:mono-props} establishes that the updates lead to a decreasing sequence of objective values $F(A_{k},B_k)$ in \eqref{eq:mmmf}.
Section~\ref{conv-rates-ch} (Theorem~\ref{thm:conv-rate-1} and Corollary~\ref{cor-1-conv-rate}) 
derives the finite-time convergence rate properties of the proposed algorithm~\texttt{softImpute-ALS}.
Section~\ref{sec:aympt-results-1} provides descriptions of the first order stationary conditions for problem~\eqref{eq:mmmf},
the fixed points of the algorithm~\texttt{softImpute-ALS} and 
the limiting behavior of the sequence $(A_{k},B_{k}), k \geq 1$ as $k \rightarrow \infty$.
Section~\ref{sec:connect-props} (Lemma~\ref{dec-seq-obj-als-1}) 
investigates the implications of the updates produced by~\texttt{softImpute-ALS} for 
problem~\eqref{eq:mmmf} in terms of the problem~\eqref{eq:1}.
Section~\ref{sec:closer-look} (Theorem~\ref{equivalence-1}) studies the stationarity conditions for problem~\eqref{eq:mmmf} vis-a-vis 
the optimality conditions for the convex problem~\eqref{eq:1}.

 The \texttt{softImpute-ALS} algorithm may be thought of as an EM- or more generally a MM-style algorithm (majorization minimization),  where every imputation step leads to an upper bound to the training error part of the loss function. 
 The resultant loss function is minimized wrt $A$---this leads to a partial minimization of the objective function wrt $A$. 
The process is repeated with the other factor $B$, and continued till convergence.

Recall the objective function of MMMF in \eqref{eq:mmmf}:
$$
F(A,B):= \frac{1}{2}\norm{P_\Omega (X-AB^T)}_F^2 +\frac{\lambda}{2}\norm{A}_F^2+\frac{\lambda}{2}\norm{B}_F^2.
$$
We define the surrogate functions
\begin{eqnarray} 
Q_A(Z_1|A,B) &:=& \frac{1}{2}\norm{P_\Omega (X-Z_1B^T)+P_\Omega^\perp (AB^T-Z_1B^T)}_F^2 \label{defn-QA}\\
&&\qquad +\frac{\lambda}{2}\norm{Z_1}_F^2+\frac{\lambda}{2}\norm{B}_F^2\\
Q_B(Z_2|A,B) &:=& \frac{1}{2}\norm{P_\Omega(X-AZ_2^T)+P_\Omega^\perp(AB^T-AZ_2^T)}_F^2\label{defn-QB}\\
&&\qquad+\frac{\lambda}{2}\norm{A}_F^2+\frac{\lambda}{2}\norm{Z_2}_F^2.
\end{eqnarray}



Consider the function $g(AB^T):= \frac{1}{2}\norm{P_\Omega (X-AB^T)}_F^2$ which is the training error as a 
function of the outer-product $Z=AB^T$, and observe that for  any $Z, \overline{Z}$ we have:
 \begin{equation}\label{major-1-1}
 \begin{aligned}
 g(Z) &\leq &  \frac{1}{2}\norm{P_\Omega (X- Z )+P_\Omega^\perp (\overline{Z} - Z)}_F^2 \\
  &=& \frac{1}{2}\norm{\left( P_\Omega (X) + P_{\Omega}^\perp(\overline{Z}) \right) - Z}_F^2
 \end{aligned}
 \end{equation}
where,  equality holds above at $Z= \overline{Z}$.
This leads to the following simple but important observations:
\begin{equation} \label{major-1-2}
Q_A(Z_1|A,B) \geq F(Z_{1}, B), \qquad Q_B(Z_2|A,B)\geq F(A, Z_{2}), 
\end{equation}
suggesting that $Q_A(Z_1|A,B)$ is a majorizer of $F(Z_{1}, B)$ (as a function of $Z_{1}$); similarly, 
$Q_B(Z_2|A,B)$ majorizes  $F(A, Z_{2})$. In addition, equality holds as follows:
\begin{equation} \label{major-1-2-eq}
Q_A(A |A,B)=F(A,B)=Q_B(B|A,B).
\end{equation}

We also define $X^*_{A,B}= P_{\Omega}(X)  +P_{\Omega}^\perp (AB^T)$.
Using these definitions, we can succinctly describe the \texttt{softImpute-ALS} algorithm in Algorithm~\ref{alg:fast-als}. This is almost equivalent to Algorithm~\ref{alg:rrsimpute}, but more convenient for theoretical analysis. It has the  orthogalization and redistribution of $\tilde D$ in step 3 removed, and step~5 removed.
\begin{algorithm}
\caption{\textbf{\texttt{softImpute-ALS}}}
\textbf{Inputs:}
Data matrix $X$, initial iterates $A_0$ and $B_0$, and $k=0$.\\
\textbf{Outputs:}
$(A^*, B^*) = \argmin_{A,B}\ F(A,B)$\\
\textbf{Repeat until Convergence}
\begin{enumerate}
\item $k\gets k+1$.
\item $X^* \gets P_{\Omega}(X)  +P_{\Omega}^\perp (AB^T)=P_\Omega ( X- AB^T) +AB^T$
\item  $A \gets X^* B(B^TB +\lambda I)^{-1}  = \argmin_{Z_1}\ Q_A (Z_1|A,B)$.
\item $X^* \gets P_{\Omega}(X)  +P_{\Omega}^\perp (AB^T)$
\item  $B \gets X^{*T} A (A^TA +\lambda I)^{-1}= \argmin_{Z_2}\ Q_B (Z_2| A,B)$.
\end{enumerate}
\label{alg:fast-als}
\end{algorithm}
Observe that the \texttt{softImpute-ALS} algorithm  can be described as the following iterative procedure:
\begin{eqnarray}
A_{k+1}  &\in& \argmin_{Z_1} \;\; Q_A(Z_1|A_{k},B_{k})  \label{update-A1}\\
B_{k+1}  &\in& \argmin_{Z_2} \;\; Q_B(Z_2|A_{k+1},B_{k}).  \label{update-B1}
\end{eqnarray}
We will use the above notation in our proof.

We can easily establish that \texttt{softImpute-ALS} is a descent method, or its iterates never increase the function value.
\begin{thm}\label{thm:mono-props}
Let $\{(A_k,B_k)\}$ be the iterates generated by {\em \texttt{softImpute-ALS}}. The function values are monotonically decreasing,
\begin{equation*}
F(A_k,B_k)\ge F(A_{k+1},B_k)\ge F(A_{k+1},B_{k+1}).
\end{equation*}
\label{thm:monotone}
\end{thm}
\begin{proof}
Let the current iterate estimates be $(A_{k}, B_{k})$.
We will first consider the update in $A$, leading to $A_{k+1}$~\eqref{update-A1}.
$$\min_{Z_{1}} \;\; Q_A(Z_1|A_{k},B_{k})   \leq  Q_A(A_{k} |A_{k},B_{k}) =  F(A_k, B_k)$$
Note that,  $\min_{Z_{1}}\; Q_A(Z_1|A_{k},B_{k})  = Q_A(A_{k+1}|A_k,B_k)$, by definition of $A_{k+1}$ in \eqref{update-A1}.

Using~\eqref{major-1-2} we get that 
$Q_A(A_{k+1}|A_k,B_k)  \geq F(A_{k+1}, B_{k}).$ Putting together the pieces we get:
$F(A_k, B_k) \ge F(A_{k+1}, B_k).$

Using an argument exactly similar to the above for the update in $B$ we have:
\begin{align}\label{nesting-1}
F(A_{k+1}, B_k) = Q_B(B_k|A_{k+1},B_k) \ge Q_B(B_{k+1}|A_{k+1},B_k) \ge F(A_{k+1}, B_{k+1}).
\end{align}
This establishes that 
$F(A_k,B_k) \ge F(A_{k+1}, B_{k+1})$ for all $k$, thereby completing the proof of the theorem.
\end{proof}

\subsection{Convergence Rates}\label{conv-rates-ch}
The previous section derives some elementary properties of the \texttt{softImpute-ALS} algorithm, namely the updates lead to a deceasing 
sequence of objective values and every limit point of the sequence is a stationary point\footnote{under minor regularity conditions, namely $\lambda>0$} of the optimization problem~\eqref{eq:mmmf}. These are all asymptotic characterizations and do not inform us about the rate at which the \texttt{softImpute-ALS} algorithm reaches a stationary point.

We present the following lemma:
\begin{lem}\label{lem:diff-fast-als1}
Let $(A_{k}, B_{k})$ denote the values of the factors at iteration $k$. We have the following: 
\begin{equation}\label{upd-als-two-1}
\begin{aligned}
F(A_k, B_k) - F(A_{k+1}, B_{k+1}) & \geq  \half \left( \| ( A_{k}   - A_{k+1}) B^T_{k} \|_F^2 + \| A_{k+1} (B_{k+1}   -  B_{k} )^T \|_F^2 \right) \\
&+ \frac{\lambda}{2} \left ( \| A_{k}   - A_{k+1}  \|_F^2  + \| B_{k+1}   -  B_{k} \|_F^2 \right)
\end{aligned}
\end{equation}
\end{lem}

For any two matrices $A$ and $B$ respectively define
$A^{+}, B^{+}$ as follows:
\begin{equation}\label{defn-update-1}
A^+ \in \argmin_{Z_{1}} \; Q_{A}( Z_{1}  | A , B), \;\; \quad \;\; B^+ \in \argmin_{Z_{2}} \; Q_{B}( Z_{2}  | A , B)
\end{equation}
We will consequently define the following:
\begin{equation}\label{defn:metric1-0}
\begin{aligned}
\Delta \left (  \left( A, B  \right ) , \left (  A^{+}, B^{+}    \right)    \right ):= \half \left( \| (A  - A^{+} )B^T\|_F^2 + \| A^{+} (B   - B^{+} )^T \|_F^2 \right) & \\
 + \frac{\lambda}{2} \left ( \| A   - A^{+}\|_F^2  + \| B   - B^{+}\|_F^2 \right)&
\end{aligned}
\end{equation}

\begin{lem}\label{lem-1-stat-1}
$\Delta \left (  \left( A, B  \right ) , \left (  A^{+}, B^{+}    \right)    \right )=0$ \emph{iff}
$A,B$ is a fixed point of {\em \texttt{softImpute-ALS}}.
\begin{proof}
Section~\ref{proof-lem-1-stat-1}.
\end{proof}
\end{lem}

We will use the following notation
\begin{equation}\label{defn:metric1}
\begin{aligned}
\eta_{k} :=& \Delta \left (  \left( A_{k}, B_{k}  \right ) , \left (  A_{k+1}, B_{k+1}    \right)    \right )
\end{aligned}
\end{equation}

Thus $\eta_{k}$ can be used to quantify how close $(A_k,B_k)$ is from a stationary point.

If $\eta_{k} >0$ it means the Algorithm will make progress in improving the quality of the solution. 
As a consequence of the monotone decreasing property of 
the sequence of objective values $F(A_k,B_k)$ and Lemma~\ref{lem:diff-fast-als1},
we have that, $\eta_{k} \rightarrow 0$ as $k \rightarrow \infty$. The following theorem characterizes the rate at which $\eta_{k}$ converges to zero.

%

\begin{thm}\label{thm:conv-rate-1}
Let $(A_{k}, B_{k}), k \geq 1$ be the sequence generated by the {\em \texttt{softImpute-ALS}} algorithm. 
The decreasing sequence of objective values
$F(A_{k}, B_{k})$ converges to $f^{\infty} \geq 0$ (say) and the quantities $\eta_{k} \rightarrow 0$.  

Furthermore, we have the  following finite convergence rate of the {\em \texttt{softImpute-ALS}} algorithm:
\begin{equation}\label{thm-rate-1}
\min_{1 \leq k \leq K } \eta_{k}  \leq  \left ( F( A_{1}, B_{1} ) - f^{\infty} ) \right) / K 
\end{equation}
\begin{proof}
See Section~\ref{proof-thm:conv-rate-1}
\end{proof}
\end{thm}

The above theorem establishes a $O(\frac{1}{K})$ convergence rate of \texttt{softImpute-ALS}; in other words, 
for any $\epsilon>0$, we need at most $K = O(\frac{1}{\epsilon})$ iterations to arrive at a point $(A_{k^*}, B_{k^*})$ such that 
 $\eta_{k^*} \leq \epsilon$, where, $1 \leq k^* \leq K$.
Note that Theorem~\ref{thm:conv-rate-1} establishes convergence of the algorithm for \emph{any} value of $\lambda \geq 0$. 
We found in our numerical experiments that the value of $\lambda$ has an important role to play in the speed of convergence of the algorithm.
In the following corollary, we provide convergence rates that make the role of $\lambda$ explicit.

The following corollary employs three different distance measures to measure the closeness of a point from stationarity.

\begin{cor}\label{cor-1-conv-rate}
Let $(A_k, B_k), k \geq 1$ be defined as above. Assume that for all $k\geq 1$
\begin{equation}\label{low-up-bound1}
 \ell^{U} \M{I} \succeq  B_{k}' B_{k}  \succeq \ell^{L}\M{I}, \;\;  \ell^{U} \M{I} \succeq  A_{k}' A_{k}   \succeq \ell^{L}\M{I},
 \end{equation}
where, $\ell^U, \ell^L$ are constants independent of $k$.

Then we have the following:

\begin{align}
 &\min_{1 \leq k \leq K } \left( \|  (A_{k}   - A_{k+1} )\|_F^2  +   \| B_{k}   -  B_{k+1}  \|_F^2 \right) \leq  
\frac{2}{(\ell^{L}+ \lambda )}\left(\frac{F( A_{1}, B_{1} ) - f^{\infty}}{K} \right) \label{bound-diff-norm-1}  \\
& \min_{1 \leq k \leq K } \left ( 
\|  (A_{k}   - A_{k+1} )B_{k}^T\|_F^2 + \| A_{k+1} ( B_{k}   -  B_{k+1})^T  \|_F^2 \right) \leq  
\frac{2 \ell^{U}}{\lambda + \ell_U}  \left(\frac{F( A_{1}, B_{1} ) - f^{\infty}}{K} \right) \label{bound-diff-norm-2} \\
& \min_{1 \leq k \leq K }  \left(  \| \nabla_{A} f (A_{k},B_{k}) \|^2 +  \| \nabla_{B} f (A_{k+1},B_{k}) \|^2  \right) \leq
\frac{2(\ell^{U})^2}{(\ell^{L}+ \lambda )}\left(\frac{F( A_{1}, B_{1} ) - f^{\infty}}{K} \right)\label{bound-diff-norm-3}
\end{align}
where, $\nabla_{A} f (A,B)$  (respectively, $\nabla_{B} f (A,B)$) denotes the partial derivative of $F(A,B)$ wrt $A$
(respectively, $B$).
\begin{proof}
See Section~\ref{proof:cor-1-conv-rate}.
\end{proof}
\end{cor}

Inequalities~\eqref{bound-diff-norm-1}--\eqref{bound-diff-norm-3} are statements 
about different notions of distances between successive iterates. These may be employed to understand the  
convergence rate of \texttt{softImpute-ALS}. 

Assumption~\eqref{low-up-bound1} is a minor one. While it may not be able to estimate $\ell^U$ prior to 
running the algorithm, a finite value of $\ell^U$ is guaranteed as soon as $\lambda >0$. 
The lower bound $\ell_{L} > 0$, if both $A_{1} \in \Re^{ m \times r}, B_{1} \in \Re^{n \times r}$ have rank $r$ and the rank 
remains the same across the iterates. If the solution to problem~\eqref{eq:mmmf} has a rank smaller than $r$, then 
$\ell_{L}=0$, however, this situation is typically avoided since a small value of $r$ leads to lower computational cost per iteration of
the \texttt{softImpute-ALS} algorithm. The constants appearing as a part of the rates 
in~\eqref{bound-diff-norm-1}--\eqref{bound-diff-norm-3} are dependent upon $\lambda$. The 
constants are smaller for larger values of $\lambda$. 
Finally we note that the algorithm does not require any information about the constants
$\ell^L, \ell^U$ appearing as a part of the rate estimates.

\subsection{Asymptotic Convergence Results}\label{sec:aympt-results-1}
In this section we derive some properties of the limiting behavior of the sequence $A_k, B_{k}$, in particular we 
examine the properties of the limit points of the sequence $A_{k}, B_{k}$.

At the beginning we recall the notion of first order stationarity of a point $A_*,B_*$.
We say that $A_*, B_* $ is said to be a first order stationary point for the problem~\eqref{eq:mmmf} if the following holds:
\begin{equation}\label{stat-point-defn}
\nabla_{A} f (A_*, B_*) = 0, \;\;\;\; \nabla_{B} f (A_*, B_*) = 0.
\end{equation}
An equivalent restatement of condition~\eqref{stat-point-defn} is:
\begin{equation}\label{stat-point-defn-1}
A_* \in  \argmin_{Z_{1}}  Q_A(Z_{1}|A_*,B_*), \;\;\;\; B_* \in  \argmin_{Z_{2}}  Q_B(Z_{2}|A_*,B_*),
\end{equation}
i.e.,  $A_*,B_*$ is a fixed point of the \texttt{softImpute-ALS} algorithm updates.

We now consider uniqueness properties of the limit points of $(A_{k}, B_{k}), k \geq 1$. 
Even in the fully observed case, the stationary points of the problem~\eqref{eq:mmmf} are not unique in $A_*,B_*$; due to orthogonal invariance. 
Addressing convergence of $(A_k,B_k)$ becomes trickier if two singular values of $A_*B_*^T$ are tied. 
In this vein we have the following result:
\begin{thm}\label{thm-glob-cov1}
Let $\{(A_k,B_k)\}_k$ be the sequence of iterates generated by Algorithm \ref{alg:fast-als}. For $\lambda>0$,  we have:
\begin{enumerate}
\item [(a)] Every limit point of $\{(A_k,B_k)\}_k$ is a stationary point of problem~\eqref{eq:mmmf}.
\item[(b)]  Let $B_*$ be any limit point of the sequence $B_{k}, k \geq 1$, with $B_{\nu} \rightarrow B_*$, where, $\nu$ is a subsequence of  $\{ 1, 2, \ldots, \}$.
Then the sequence $A_{\nu}$ converges. 
\end{enumerate}
Similarly, let $A_*$ be any limit point of the sequence $A_{k}, k \geq 1$, with $A_{\mu} \rightarrow B_*$, where, $\mu$ is a subsequence of  $\{ 1, 2, \ldots, \}$.
Then the sequence $A_{\mu}$ converges. 

\begin{proof}
See Section~\ref{proof:thm-glob-cov1}
\end{proof}
\end{thm}
The above theorem is a partial result about the uniqueness of the limit points of the sequence $A_{k},B_{k}$.
The theorem implies that if the sequence $A_k$ converges, then the sequence $B_{k}$ must converge and vice-versa.
More generally, for every limit point of $A_k$, the associated $B_{k}$ (sub)sequence will converge. The same result holds true for the 
sequence $B_{k}$.

\begin{remark}\label{rem-1}
Note that the condition $\lambda>0$ is enforced due to technical reasons so that the sequence $(A_{k}, B_{k})$ remains bounded. If $\lambda=0$, then 
$A \leftarrow c A$ and $B \leftarrow \frac1c B$ for any $c>0$, leaves the objective function unchanged.  Thus one may take 
$c \rightarrow \infty$ making the sequence of updates unbounded without making any change to the values of the objective function.
 \end{remark}

%
%


\subsection{Properties of the \texttt{softImpute-ALS} updates in terms of Problem~\eqref{eq:mmmf}}\label{sec:connect-props}



The sequence $(A_{k}, B_{k})$ generated by Algorithm~(\ref{alg:fast-als}) are geared towards minimizing criterion~\eqref{eq:mmmf}, 
it is not obvious as to what the implications the sequence has for the convex problem~\eqref{eq:1}
In particular, we know that $F(A_{k}, B_{k})$ is decreasing---does this imply a monotone sequence $H(A_{k}B_{k}')$?
We show below that it is indeed possible to obtain a monotone decreasing sequence $H(A_{k}B_{k}')$ with a minor modification.
These modifications are exactly those implemented in Algorithm~\ref{alg:rrsimpute} in step~3.




The idea that plays a crucial role in this modification is the following 
inequality (for a proof see~\cite{mazumder09:_spect_regul_algor_for_learn}; see also remark~\ref{item:3} in Section~\ref{eMMMF}):
$$ \| AB^T\|_* \leq \half ( \| A \|_F^2 + \|B \|_F^2 ).$$
Note that equality holds above if we take a particular choice of $A$ and $B$ given by:
\begin{equation}\label{svd-hold-1}
 A=UD^{1/2}, B = VD^{1/2}, \;\;\;\; \text{where,}\;\;\;\; AB' = UDV^T\;\;\;\; (\text{SVD}),
 \end{equation}
is the SVD of $AB^T.$
The above observation implies that if $(A_{k}, B_{k})$ is generated by \texttt{softImpute-ALS} then 
$$F( A_{k}, B_{k}) \geq H(A_{k}B_{k}^T)$$
with equality holding if $A_{k},B_{k}^T$ are represented via~\eqref{svd-hold-1}. Note that this re-parametrization 
does not change the training error portion of the objective $F( A_{k}, B_{k})$, but decreases the ridge regularization term---and hence 
decreases the overall objective value when compared to that achieved by~\texttt{softImpute-ALS} without the reparametrization~\eqref{svd-hold-1}.

We thus have the following Lemma:
\begin{lem}\label{dec-seq-obj-als-1}
Let the sequence $(A_{k}, B_{k})$ generated by \texttt{softImpute-ALS}  be stored in terms of the factored SVD representation~\eqref{svd-hold-1}.
This results in a decreasing sequence of objective values in the nuclear norm penalized problem~\eqref{eq:1}:
$$ H(A_{k}B_{k}^T)  \geq  H(A_{k+1}B_{k+1}^T)$$
with $H(A_{k}B_{k}^T)= F(A_{k}, B_{k} )$, for all $k$.
The sequence $H(A_{k}B_{k}^T)$ thus converges to $f^{\infty}.$
\end{lem}

Note that, $f^{\infty}$ need not be the minimum of the convex problem~\eqref{eq:1}. It is easy to see this, by taking $r$ to be smaller 
than the rank of the optimal solution to problem~\eqref{eq:1}.

\subsubsection{A Closer Look at the Stationary Conditions}\label{sec:closer-look}
In this section we inspect the first order stationary conditions of the non-convex problem~\eqref{eq:mmmf} and the 
convex problem~\eqref{eq:1}. We will see that a first order stationary point of the convex problem~\eqref{eq:1} leads to 
factors $(A,B)$ that are stationary for problem~\eqref{eq:mmmf}. However, the converse of this statement need not be true in general. However, given an estimate delivered by 
\texttt{softImpute-ALS} (upon convergence) it is easy to verify whether it is a solution to problem~\eqref{eq:1}.

Note  that $Z^*$ is the optimal solution to the convex problem~\eqref{eq:1} iff:
$$ \partial H(Z^*) = P_{\Omega}(Z^* - X)  + \lambda \sgn( Z^* ) = 0,$$
where, $\sgn(Z^*)$ is a sub-gradient of the nuclear norm $\|Z\|_*$ at $Z^*$. 
Using the standard characterization~\cite{lewis-96} of $\sgn( Z^* )$ the above condition is equivalent to: 
\begin{equation}\label{stat-cvx-sub1}
P_{\Omega}(Z^* - X)  + \lambda U_{*}\sgn(D^{*})V^T_{*} = 0   
\end{equation}
where, the full SVD of  $Z^*$ is given by $U_{*}D_{*}V_{*}^T$;
 $\sgn(D^{*})$ is a diagonal matrix with $i$th diagonal entry given by 
$\sgn(d^*_{ii}),$ where, $d_{ii}^*$ is the $i$th diagonal entry of $D^*$.

If a limit point $A_*B_*^T$ of the \texttt{softImpute-ALS} algorithm
satisfies the stationarity condition~\eqref{stat-cvx-sub1} above, then it is the optimal solution of the convex problem.
We note that $A_*B_*^T$ need not necessarily satisfy the stationarity condition~\eqref{stat-cvx-sub1}.

$(A,B)$ satisfy the stationarity conditions of \texttt{softImpute-ALS} if the following conditions are satisfied:
$$P_{\Omega}(AB^T - X) B  + \lambda A = 0 , \;\;\;   A^T(P_{\Omega}(AB^T - X))  + \lambda B^T = 0,$$
where, we assume that $A,B$ are represented in terms of~\eqref{svd-hold-1}. This gives us:
\begin{equation}\label{stat-fast-als-1}
P_{\Omega}(AB^T - X) V  + \lambda U = 0 , \;\;\;   U^T(P_{\Omega}(AB^T - X))  + \lambda V^T = 0,
\end{equation}
where $AB^T = UDV^T$, being the reduced rank SVD i.e. all diagonal entries of $D$ are strictly positive. 

A stationary point of the convex problem corresponds to a stationary point of \texttt{softImpute-ALS}, as seen by a direct verification of the conditions above. 
In the following we investigate the converse: when does a stationary point of \texttt{softImpute-ALS} correspond to a stationary point of ~\eqref{eq:1}; i.e. condition~\eqref{stat-cvx-sub1}?
Towards this end, we make use of the ridged least-squares update used by \texttt{softImpute-ALS}.
Assume that all matrices $A_{k},B_{k}$ have $r$ rows.

At stationarity i.e. at a fixed point of \texttt{softImpute-ALS} we have the following:
\begin{align}
A_* &\in& \argmin_{A}& \;\; \half \|P_{\Omega}(X - AB^T_{*}) \|_F^2  + \frac{\lambda}{2} \left ( \|A\|_F^2 + \|B_*\|_F^2   \right) 
\label{stat-1-1-1} \\
&=& \argmin_{A}& \;\; \half \| \left ( P_{\Omega}(X) +  P^\perp_{\Omega}(A_*B^T_*) \right) - AB^T_* \|_F^2  + \frac{\lambda}{2} \left ( \|A\|_F^2 + \|B_*\|_F^2   \right) \label{stat-1-1-2} \\
B_* &\in& \argmin_{B}& \;\; \half \|P_{\Omega}(X - A_*B^T ) \|_F^2  + \frac{\lambda}{2} \left ( \|A_*\|_F^2 + \|B\|_F^2   \right)
\label{stat-1-2-1} \\
&=& \argmin_{B}& \;\; \half \| \left ( P_{\Omega}(X) +  P^\perp_{\Omega}(A_*B_*^T) \right) - A_*B^T \|_F^2  + \frac{\lambda}{2} \left ( \|A_*\|_F^2 + \|B\|_F^2   \right) \label{stat-1-2-2} 
\end{align}
Line~\eqref{stat-1-1-2} and~\eqref{stat-1-2-2} can be thought of doing alternating multiple ridge regressions for the fully observed matrix 
$P_{\Omega}(X) +  P^\perp_{\Omega}(A_*B^T_*)$. 

The above fixed point updates are very closely related to the following optimization problem:
\begin{equation}\label{stat-nuke-norm-ridging-1}
\mini_{A_{m\times r},B_{m \times r}} \;\; \half \| \left ( P_{\Omega}(X) +  P^\perp_{\Omega}(A_*B_*^T) \right) - AB \|_F^2  + \frac{\lambda}{2} \left ( \|A\|_F^2 + \|B\|_F^2   \right)  
\end{equation}
The solution to~\eqref{stat-nuke-norm-ridging-1} by Theorem~\ref{thm1} is given by the nuclear norm thresholding operation (with a rank $r$ constraint) on the matrix $P_{\Omega}(X) +  P^\perp_{\Omega}(A_*B^T_*)$:
\begin{equation}
  \label{eq:rrsoftimpute}
\mini_{Z: \rnk(Z) \leq r} \; \half \| \left ( P_{\Omega}(X) +  P^\perp_{\Omega}(A_*B_*^T) \right) - Z \|_F^2 +  \frac{\lambda}{2} \| Z\|_{*}.
\end{equation}

Suppose the convex optimization problem~\eqref{eq:1} has a solution $Z^*$ with $\mathrm{rank}(Z^*)= r^*$. Then, for 
$A_*B_*^T$ to be a solution to the convex problem the following conditions are sufficient:
\begin{enumerate}
\item[(a)] $r^* \leq r$
\item[(b)] $A_*, B_*$ must be the global minimum of problem~\eqref{stat-nuke-norm-ridging-1}. 
Equivalently, the outer product 
$A_*B_*^T$ must be the solution to the fully observed nuclear norm problem:
\begin{equation}\label{equi-cond-1}
A_*B_*^T  \in \argmin_{Z}\;\;  \half \| \left ( P_{\Omega}(X) +  P^\perp_{\Omega}(A_*B_*^T) \right) - Z \|_F^2  +  \lambda\|Z\|_*
\end{equation}
\end{enumerate}
The above condition~\eqref{equi-cond-1} can be verified and requires doing a low-rank SVD of the matrix 
$P_{\Omega}(X) +  P^\perp_{\Omega}(A_*B_*^T)$ and is a direct application of Algorithm~\ref{alg:rrssvd}.
This task is computationally attractive due to the sparse plus low-rank structure of the matrix.

\begin{thm}\label{equivalence-1}
Let $A_{k} \in \Re^{m \times r}, B_{k} \in \Re^{n \times r}$ be the sequence generated by {\em \texttt{softImpute-ALS}} and let $(A_*,B_*)$  denote a limit point of the sequence. 
Suppose the Problem~\eqref{eq:1} has a minimizer with rank at most $r$.
If  $Z_* = A_*B_*^T$ solves the problem~\eqref{equi-cond-1} then $Z_*$ is a solution to the convex Problem~\eqref{eq:1}.
\end{thm}

\subsection{Computational Complexity and Comparison to ALS}
The computational cost of \texttt{softImpute-ALS} can be broken down into three steps. First consider only the cost of the update to $A$. The first step is forming the matrix $X^*=P_\Omega (X - AB^T) + AB^T$, which requires $O(r|\Omega|)$ flops for the $P_\Omega (AB^T)$ part, while the second part is never explicitly formed. The matrix $B(B^TB +\lambda I)^{-1}$ requires $O(2nr^2+r^3)$ flops to form; although we keep it in SVD factored form, the coast is the same. The multiplication $X^* B(B^TB +\lambda I)^{-1}$ requires $O(r|\Omega| +mr^2 +nr^2)$ flops, using the sparse plus low-rank structure of $X^*$. The total cost of an iteration is $O(2r |\Omega|+mr^2+3nr^2 +r^3)$. 

As mentioned in Section~\ref{sec:introduction},  alternating least squares (ALS)  is a  popular algorithm for solving the matrix factorization problem in Equation~\eqref{eq:mmmf}; see Algorithm~\ref{alg:regular-als}. The \texttt{ALS} algorithm is an instance of block coordinate descent applied to \eqref{eq:mmmf}. 

The updates for ALS are given by 
\begin{align}
A_{k+1} =\argmin_A F(A,B_k)\\
B_{k+1}=\argmin_B F(A_k,B),
\end{align}
and each row of $A$ and $B$ can be computed via a separate ridge regression. The cost for each ridge regression is $O(|\Omega_j| r^2 +r^3)$, so the cost of one iteration is $O(2|\Omega |r^2 +mr^3 +nr^3)$. Hence the cost of one iteration of \texttt{ALS} is $r$ times more flops than one iteration of \texttt{softImpute-ALS}.
\begin{algorithm}
\caption{\textbf{Alternating least squares \texttt{ALS}}}
\textbf{Inputs:}
Data matrix $X$, initial iterates $A_0$ and $B_0$, and $k=0$.\\
\textbf{Outputs:}
$(A^*, B^*) = \argmin_{A,B}\ F(A,B)$\\
\textbf{Repeat until Convergence}
\begin{algorithmic}
\FOR{i=1 to m}
 \STATE{$A_{i} \gets \left(\sum_{j \in \Omega_i} B_j B_j ^T\right)^{-1} \left( \sum_{j\in \Omega_i} X_{ij} B_j\right)$   } 
 \ENDFOR
 \FOR{j=1 to n}
 \STATE{$B_{j} \gets \left(\sum_{i \in \Omega_j} A_i A_i ^T\right)^{-1} \left( \sum_{i\in \Omega_j} X_{ij} A_i\right)$   } 
 \ENDFOR
\end{algorithmic}
\label{alg:regular-als}
\end{algorithm}
We will see in the next sections that while \texttt{ALS} may decrease the criterion at each iteration more than \texttt{softImpute=ALS}, it tends to be slower because the cost is higher by a factor~$O(r)$.

\section{Experiments}
In this section we run some timing experiments on simulated and real datasets, and show performance results on the Netflix and MovieLens data.

\subsection{Timing experiments}
\label{sec:timing}
Figure~\ref{fig:timing} shows timing results on four datasets. The first three are simulation datasets of increasing size, and the last is the publicly available {\tt MovieLens 100K} data. These experiments were all run in {\tt R} using the {\tt softImpute} package; see Section~\ref{sec:r-package-tt}. Three methods are compared:
\begin{enumerate}
\item {\tt ALS}--- Alternating Least Squares as in Algorithm~\ref{alg:regular-als};
\item {\tt softImpute-ALS} --- our new approach, as defined in Algorithm~\ref{alg:rrsimpute} or \ref{alg:fast-als};
\item {\tt softImpute} --- the original algorithm of \citet{mazumder09:_spect_regul_algor_for_learn}, as layed out in steps~(\ref{eq:xhatstep})--(\ref{eq:softSVD2}).
\end{enumerate}
\begin{figure}[hbtp]
  \centerline{
    \includegraphics[width=.5\textwidth]{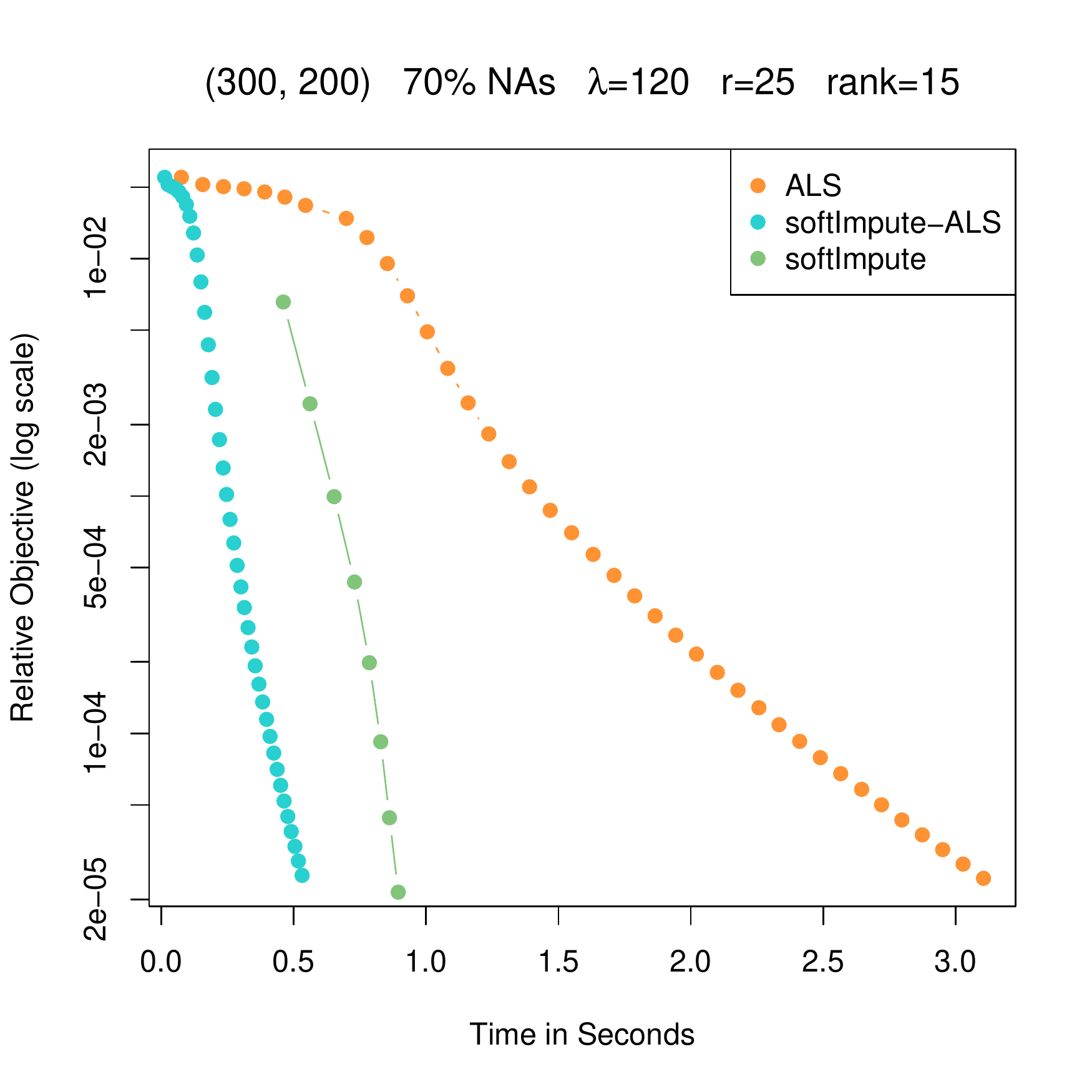}
    \includegraphics[width=.5\textwidth]{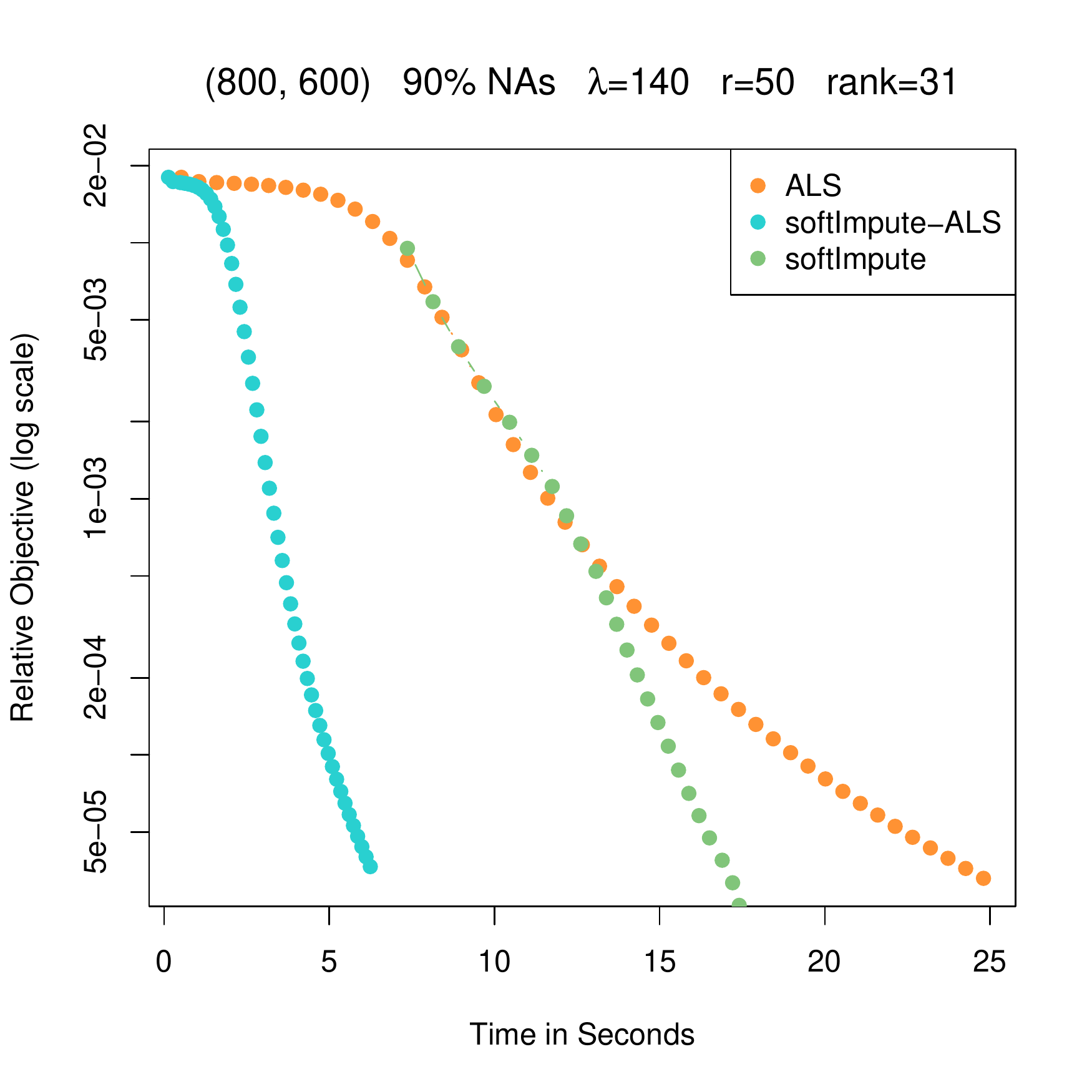}
}
 \centerline{
    \includegraphics[width=.5\textwidth]{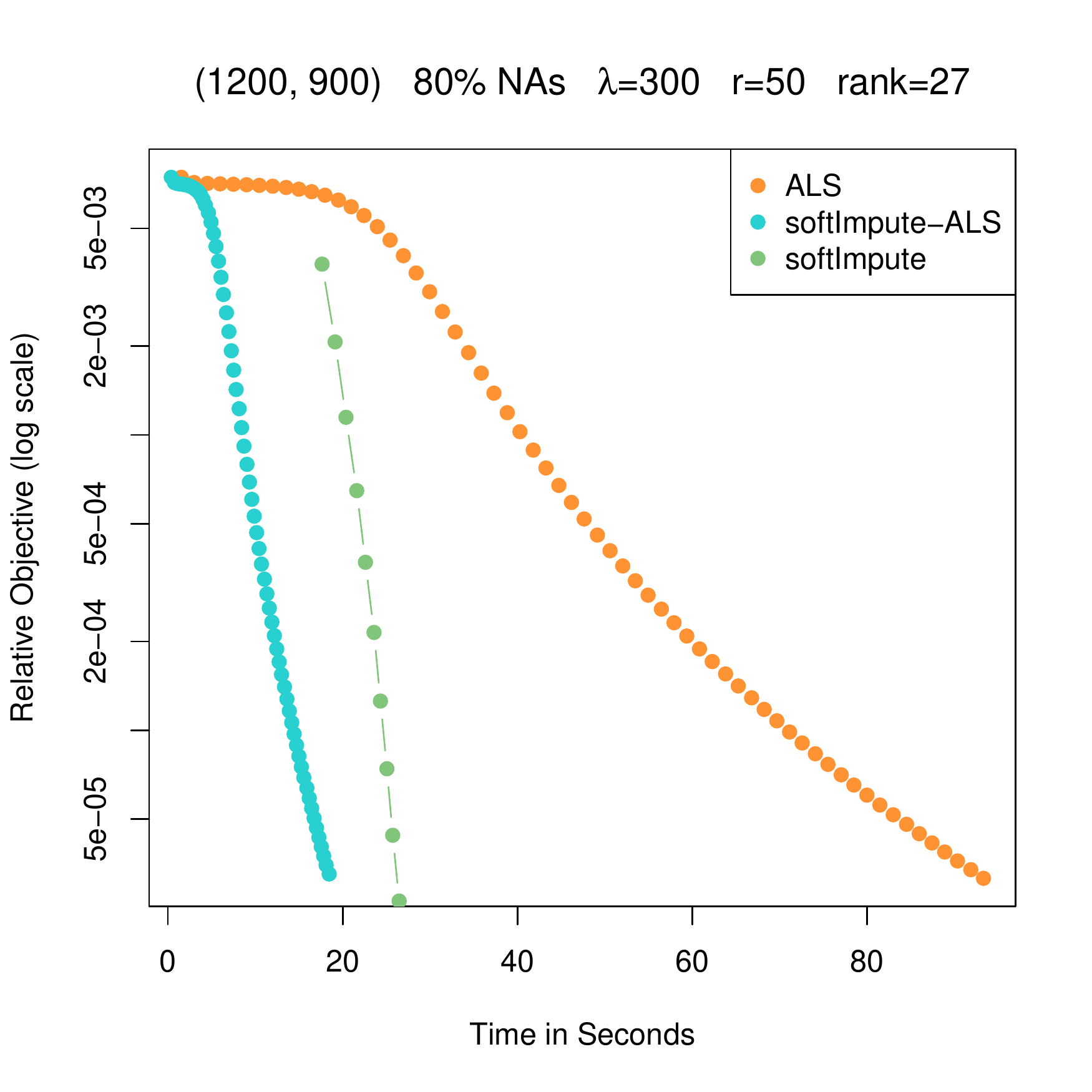}
    \includegraphics[width=.5\textwidth]{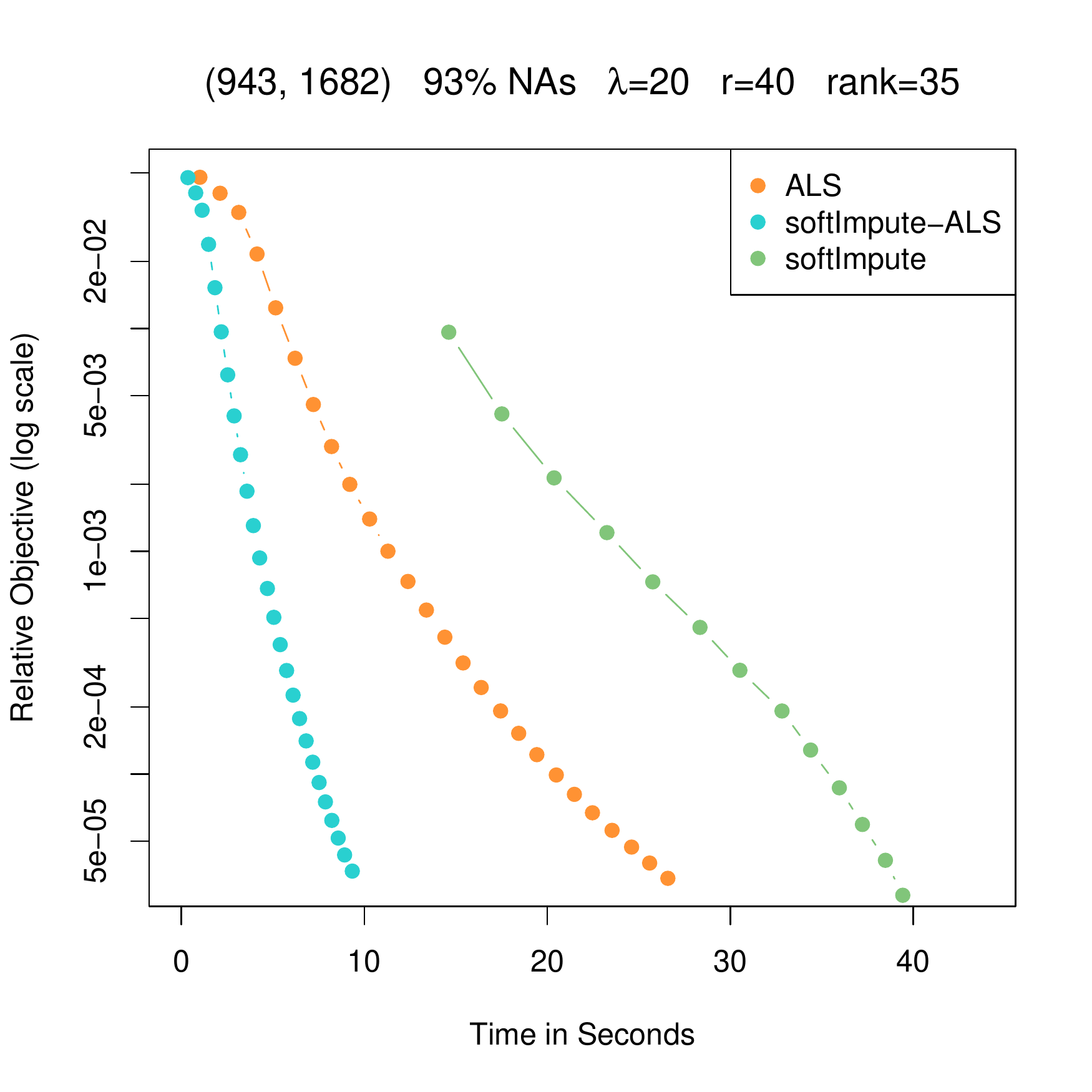}
}
\caption{Four timing experiments. Each figure is labelled according to size ($m\times n$), percentage of missing entries (NAs), value of $\lambda$ used, rank $r$ used in the ALS iterations, and rank of solution found. The first three are simulation examples, with increasing dimension. The last is the {\tt movielens 100K} data. In all cases, {\tt softImpute-ALS} (blue) wins handily against {\tt ALS} (orange) and {\tt softImpute} (green).} 
\label{fig:timing}
\end{figure}

We used an {\tt R} implementation for each of these in order to make
the fairest comparisons. In particular, algorithm {\tt softImpute}
requires a low-rank SVD of a complete matrix at each iteration. For
this we used the function {\tt svd.als} from our package, which uses
alternating subspace iterations, rather than using other optimized
code that is available for this task. Likewise, there exists optimized
code for regular {\tt ALS} for matrix completion, but instead we used
our \texttt{R} version to make the comparisons fairer. We are trying to
determine how the computational trade-offs play off, and thus need a
level playing field.

Each subplot in Figure~\ref{sec:timing} is labeled according to the
size of the problem, the fraction missing, the value of $\lambda$
used, the operating rank of the algorithms $r$, and the rank of the
solution obtained. All three methods involve alternating subspace
methods; the first two are alternating ridge regressions, and the
third alternating orthogonal regressions. These are conducted at the
operating rank $r$, anticipating a solution of smaller rank. Upon
convergence, {\tt softImpute-ALS} performs step~(5) in
Algorithm~\ref{alg:rrsimpute}, which can truncate the rank of the
solution. Our implementation of {\tt ALS} does the same.

For the three simulation examples, the data are generated from an underlying Gaussian factor model, with true ranks 50, 100, 100; the missing entries are then chosen at random. Their sizes are $(300,200)$, $(800,600)$ and $(1200,900)$ respectively, with between 70--90\% missing. The {\tt MovieLens 100K} data has 100K ratings (1--5) for 943 users and 1682 movies, and hence is 93\% missing.

We picked a value of $\lambda$ for each of these examples (through trial and error) so that the final solution had rank less than the operating rank. Under these circumstances, the solution to the criterion (\ref{eq:mmmf}) coincides with the solution to (\ref{eq:1}), which is unique under non-degenerate situations.

There is a fairly consistent message from each of these experiments. {\tt softImpute-ALS} wins handily in each case, and the reasons are clear:
\begin{itemize}
\item Even though it uses more iterations than {\tt ALS}, they are much cheaper to execute (by a factor $O(r)$).
\item {\tt softImpute} wastes time on its early SVD, even though it is far from the solution. Thereafter it uses warm starts for its SVD calculations, which speeds each step up, but it does not catch up.
\end{itemize}

\subsection{Netflix Competition Data}
\label{sec:netflix}
We used our {\tt softImpute} package  in {\tt R} to fit a sequence of models on the Netflix competition data. Here there are 480,189 users, 17,770 movies and a total of 100,480,507 ratings, making the resulting matrix 98.8\% missing. There is a designated test set (the ``probe set''), a subset of 1,408,395 of the these ratings, leaving 99,072,112 for training.
\begin{figure}[hbt]
  \centering
\includegraphics[width=\textwidth]{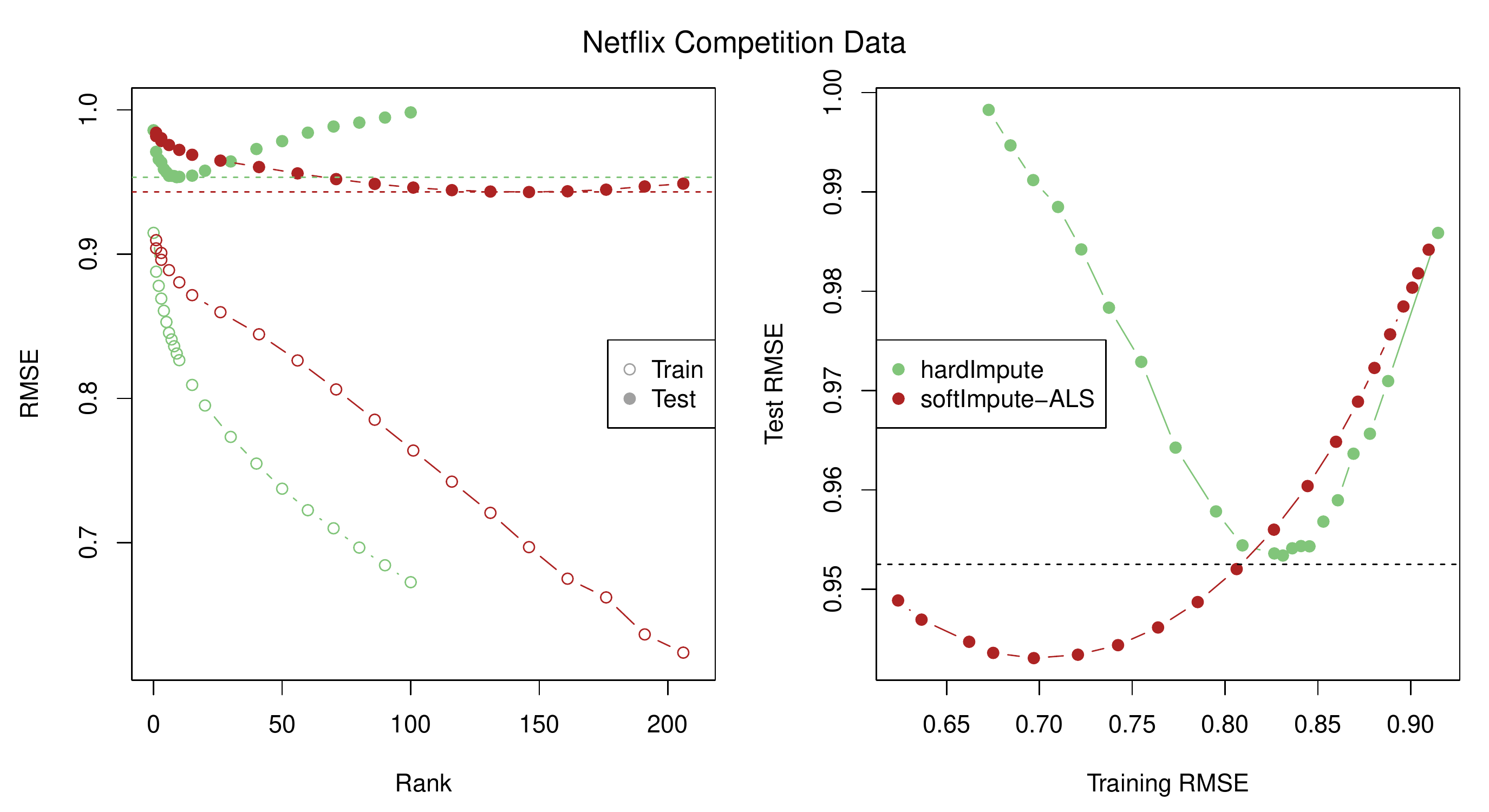}
\caption{Performance of {\tt hardImpute} versus {\tt softImpute-ALS} on the Netflix data. {\tt hardImpute} uses a rank-restricted SVD at each step of the imputation, while {\tt softImpute-ALS} does shrinking as well. The left panel shows the training and test error as a function of the rank of the solution---an imperfect calibration in light of the shrinkage. The right panel gives the test error as a function of the training error. {\tt hardImpute} fits more aggressively, and overfits far sooner than {\tt softImpute-ALS}. The horizontal dotted line is the ``Cinematch'' score, the target to beat in this competition.}
  \label{fig:netflix}
\end{figure}

Figure~\ref{fig:netflix} compares the performance of {\tt hardImpute} \cite{mazumder09:_spect_regul_algor_for_learn} with {\tt softImpute-ALS} on these data. {\tt hardImpute} uses rank-restricted SVDs iteratively to estimate the missing data, similar to {\tt softImpute} but without shrinkage. The shrinkage helps here, leading to a best test-set RMSE of 0.943. This is a 1\% improvement over the ``Cinematch'' score, somewhat short of the prize-winning improvement of 10\%.

Both methods benefit greatly from using warm starts. {\tt hardImpute} is solving a non-convex problem, while the intention is for  \texttt{softImpute-ALS} to solve the convex problem (\ref{eq:1}). This will be achieved if the operating rank is sufficiently large. The idea is to decide on a decreasing sequence of values for $\lambda$, starting from $\lambda_{max}$ (the smallest value for which the solution $\widehat M=0$, which corresponds to the largest singular value of $P_\Omega(X)$). 
Then for each value of $\lambda$, use an operating rank somewhat larger than the rank of the previous solution, with the goal of getting the solution rank smaller than the operating rank. The sequence of twenty models took under six hours of computing on a Linux cluster with 300Gb of ram (with a fairly liberal relative convergence criterion of 0.001), using the \texttt{softImpute} package in \texttt{R}.

\begin{figure}[hbtp]
  \centering
  \includegraphics[width=\textwidth]{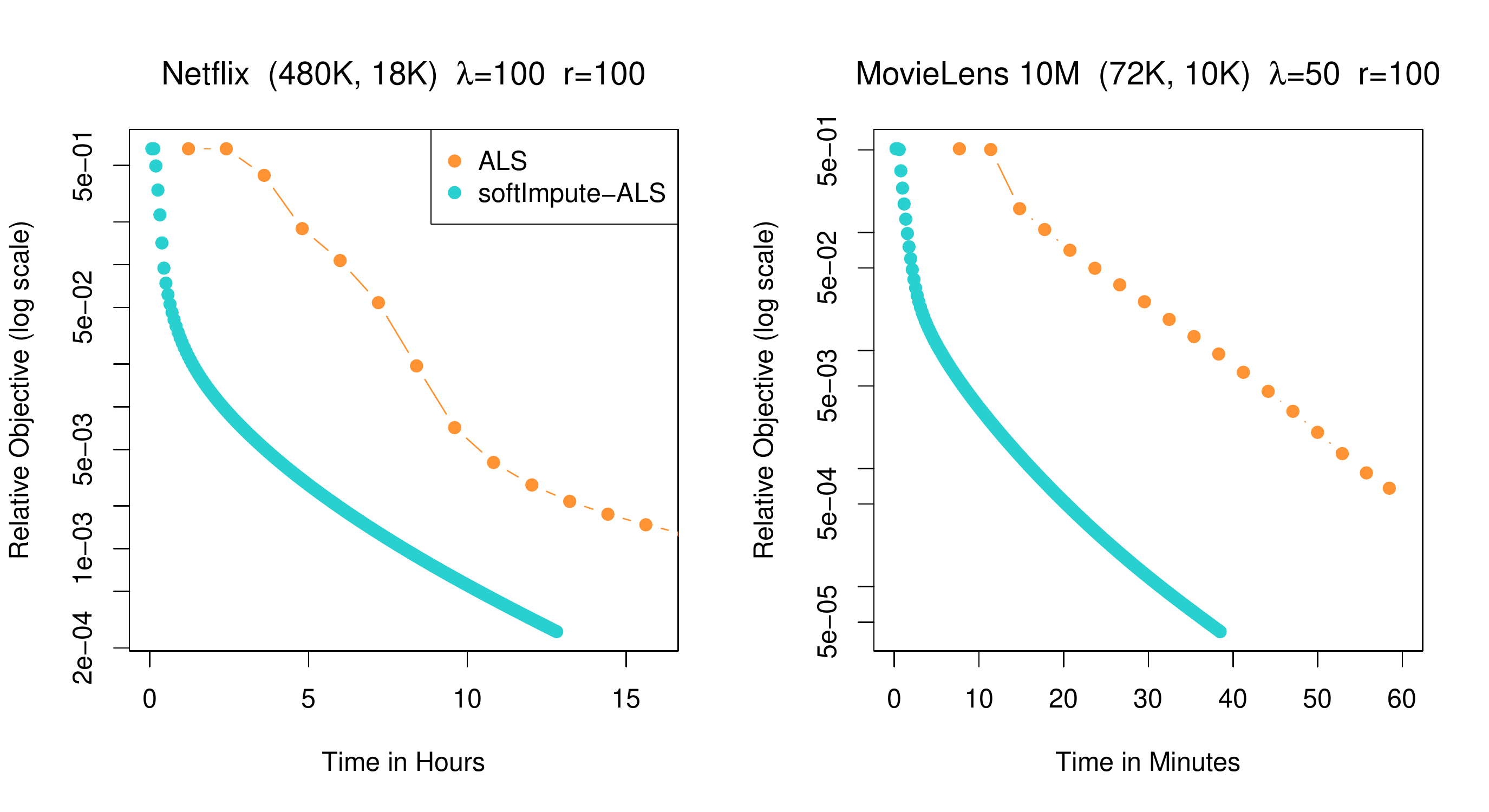}
\caption{Left: timing results on the Netflix matrix, comparing {\tt ALS} with {\tt softImpute-ALS}. Right: timing on the MovieLens 10M matrix.  In both cases we see that while {\tt ALS} makes bigger gains per iteration, each iteration is much more costly.}
\label{fig:movielensnetflix}
\end{figure}
Figure~\ref{fig:movielensnetflix} (left panel) gives timing comparison results for one of the Netflix fits, this time implemented in Matlab. The right panel gives timing results on the smaller MovieLens 10M matrix. In these applications we need not get a very accurate solution, and so early stopping is an attractive option. {\tt softImpute-ALS} reaches a solution close to the minimum in about 1/4 the time it takes {\tt ALS}.

\section{R Package {\tt softImpute}}\label{sec:r-package-tt}
We have developed an {\tt R} package {\tt softImpute} for fitting
these models \citep{softImpute_package}, which is available on
CRAN. The package implements both {\tt softImpute} and {\tt
  softImpute-ALS}. It can accommodate large matrices if the number of
missing entries is correspondingly large, by making use of
sparse-matrix formats. There are functions for centering and scaling
(see Section~\ref{centerscale}), and for making predictions from a
fitted model. The package also has a function \texttt{svd.als} for
computing a low-rank SVD of a large sparse matrix, with row and/or
column centering. More details can be found in the package Vignette on
the first authors webpage, at \\{\small\tt
  http://web.stanford.edu/~hastie/swData/softImpute/vignette.html}.

\section{Distributed Implementation}
\label{sec:distr-impl}
\subsection{Design}
We provide a distributed version of \texttt{softimpute-ALS} (given in Algorithm \ref{alg:fast-als}), built upon the Spark cluster programming framework. The input matrix to be factored is split row-by-row across many machines. The transpose of the input is also split row-by-row across the machines. The current model (i.e. the current guess for $A, B$) is repeated and held in memory on every machine. Thus the total time taken by the computation is proportional to the number of non-zeros divided by the number of CPU cores, with the restriction that the model should fit in memory.

At every iteration, the current model is broadcast to all machines, such that there is only one copy of the model on each machine. Each CPU core on a machine will process a partition of the input matrix, using the local copy of the model available. This means that even though one machine can have many cores acting on a subset of the input data, all those cores can share the same local copy of the model, thus saving RAM. This saving is especially pronounced on machines with many cores.

The implementation is available online at \url{http://git.io/sparkfastals} with documentation, in Scala. The implementation has a method named \texttt{multByXstar}, corresponding to line 3 of Algorithm \ref{alg:fast-als} which multiplies  $X^*$ by another matrix on the right, exploiting the ``sparse-plus-low-rank'' structure of $X^*$. This  method has signature:
\begin{center}
\texttt{multByXstar(X: IndexedRowMatrix, A: BDM[Double], B: BDM[Double], C: BDM[Double])}
\end{center}
This method has four parameters. The first parameter \texttt{X} is a distributed matrix consisting of the input, split row-wise across machines. The full documentation
for how this matrix is spread across machines is available online\footnote{\url{https://spark.apache.org/docs/latest/mllib-basics.html#indexedrowmatrix}}. The \texttt{multByXstar} method takes a distributed matrix, along with local matrices \texttt{A}, \texttt{B}, and \texttt{C}, and performs line 3 of Algorithm \ref{alg:fast-als} by multiplying $X^*$ by \texttt{C}. Similarly, the method \texttt{multByXstarTranspose} performs line 5 of Algorithm \ref{alg:fast-als}. 

After each call to \texttt{multByXstar}, the machines each will have calculated a portion of $A$. Once the call finishes, the machines each send their
computed portion (which is small and can fit in memory on a single machine, since $A$ can fit in memory on a single machine) to the master node, which will assemble the new guess for $A$ and broadcast it to the worker machines. A similar process happens for \texttt{multByXstarTranspose}, and the whole process is repeated every iteration.

\subsection{Experiments}
We report iteration times using an Amazon EC2 cluster with 10 slaves and one master, of instance type ``\texttt{c3.4xlarge}". Each machine has 16 CPU cores and 30 GB of RAM. We ran \texttt{softimpute-ALS} on matrices of varying sizes with iteration runtimes available in Table \ref{sparktable}, setting $k=5$. Where possible, hardware acceleration was used for local linear algebraic operations, via breeze and BLAS.

The popular Netflix prize matrix  has $17,770$ rows, $480,189$ columns, and $100,480,507$ non-zeros. We report results on several larger matrices in Table~\ref{sparktable}, up to 10 times larger.

\begin{table}[hbt]
\begin{center}
  \begin{tabular}{ | c | c | c | }
    \hline
    \textbf{Matrix Size} & \textbf{Number of Nonzeros} & \textbf{Time per iteration (s)} \\ \hline
    $10^6 \times 10^6$  & $10^6$ & 5 \\ \hline
    $10^6 \times 10^6$  & $10^9$ & 6 \\ \hline
    $10^7 \times 10^7$  & $10^{9}$ &  139\\ \hline
  \end{tabular}
  \caption[Distributed \texttt{softimpute-ALS}]{Running times for distributed \texttt{softimpute-ALS}}
   \label{sparktable}
\end{center}
\end{table}


\section{Centering and Scaling}

\label{centerscale}
We often want to remove row and/or column means from a matrix before performing a low-rank SVD or running our matrix completion algorithms. Likewise we may wish to standardize the rows and or columns to have unit variance. In this section we 
present an algorithm for doing this, in a way that is sensitive to the storage requirement of very large, sparse matrices. We first present our approach, and then discuss implementation details.

We have a two-dimensional array $X=\{X_{ij}\}\in \real^{m\times n}$, with pairs
$(i,j)\in\Omega$ observed and the rest missing. The goal is to
standardize the rows and columns of $X$ to mean zero and variance~one
simultaneously.  We consider the mean/variance model
\begin{equation}
  \label{eq:one}
  X_{ij}\sim (\mu_{ij},\sigma^2_{ij})
\end{equation}
with 
\begin{eqnarray}
  \mu_{ij}&=&\alpha_i+\beta_j;  \label{eq:twoa}\\
\sigma_{ij}&=&\tau_i\gamma_j.  \label{eq:twob}
\end{eqnarray}
Given the parameters of this model, we would standardized each observation via
\begin{eqnarray}
  \tilde X_{ij}&=&\frac {X_{ij}-\mu_{ij}}{\sigma_{ij}}\nonumber\\
&=& \frac {X_{ij}-\alpha_i-\beta_j}{\tau_i\gamma_j}.  \label{eq:three}
\end{eqnarray}

If model~(\ref{eq:one}) were correct, then each entry of the
standardized matrix, viewed as a realization of a random variable,
would have population mean/variance $(0,1)$. A consequence would be
that realized rows and columns would also have means and variances
with expected values zero and one respectively. However, we would like the observed
data to have these row and column properties.


Our representation (\ref{eq:twoa})--(\ref{eq:twob}) is not unique, but is easily fixed to be so. We can include
a constant $\mu_0$ in (\ref{eq:twoa}) and then have $\alpha_i$ and
$\beta_j$ average 0. Likewise, we can have an overall scaling
$\sigma_0$, and then have $\log\tau_i$ and $\log\gamma_j$ average
0. Since this is not an issue for us, we suppress this refinement.

We are not the first to attempt this dual centering and
scaling. Indeed, \citet{olshen10:_succes_normal_of_rectan_array}
implement a very similar algorithm for complete data, and discuss
convergence issues. Our algorithm differs in two simple ways: it
allows for missing data, and it learns the parameters of the
centering/scaling model (\ref{eq:three}) (rather than just applying
them).  This latter feature will be important for us in our
matrix-completion applications; once we have estimated the missing
entries in the standardized matrix $\widetilde X$, we will want to
{\em reverse} the centering and scaling on our predictions.

In matrix notation we can write our model
\begin{equation}
  \label{eq:matrix}
  \widetilde{\bX}={\bf D}_\tau^{-1}({\bf X}-{\boldsymbol \alpha}{\bf 1}^T-{\bf 1}{\boldsymbol \beta}^T){\bf D}_\gamma^{-1},
\end{equation}
where ${\bf D}_\tau=\mbox{diag}(\tau_1,\tau_2,\ldots,\tau_m)$, similar for ${\bf D}_\gamma$, and the missing values are represented in the full matrix as {\tt NA}s (e.g. as in {\tt R}). Although it is not the focus of this paper, this centering model is also useful for large, complete, sparse matrices $\bX$ (with many zeros, stored in sparse-matrix format). Centering would destroy the sparsity, but from (\ref{eq:matrix}) we can see we can store it in ``sparse-plus-low-rank'' format. Such a matrix can be left and right-multiplied easily, and hence is ideal for alternating subspace methods for computing a low-rank SVD. The function \texttt{svd.als} in the \texttt{softImpute} package (section~\ref{sec:r-package-tt}) can accommodate such structure.

\subsection{Method-of-moments Algorithm}\label{sec:meth-moments-algor}
We now present an algorithm for estimating the parameters.
The idea is to write down four systems of estimating equations
that demand that the transformed observed data have row
means zero and variances one, and likewise for the columns. We then iteratively solve these equations,
until all four conditions are satisfied simultaneously. We do not in general have
any guarantees that this algorithm will always converge except in the noted special cases, but
empirically we typically see rapid convergence.

Consider the estimating equation for the row-means condition (for each row~$i$)
\begin{eqnarray}
  \label{eq:rowmean}
  \frac1{n_i}\sum_{j\in \Omega_i}\widetilde X_{ij}&=&\frac1{n_i}\sum_{j\in \Omega_i}\frac{X_{ij}-\alpha_i-\beta_j}{\tau_i\gamma_j}\\
&=&0\nonumber,
\end{eqnarray}
where $\Omega_i=\{j| (i,j)\in \Omega\}$, and $n_i=|\Omega_i|\leq n$.
Rearranging we get
\begin{equation}
  \label{eq:alpha}
  \alpha_i=\frac{\sum_{j\in\Omega_i}\frac1{\gamma_j}(X_{ij}-\beta_j)}{\sum_{j\in\Omega_i}\frac1{\gamma_j}},\quad i=1,\ldots,m.
\end{equation}
This is a weighted mean of the partial residuals $X_{ij}-\beta_j$ with weights inversely proportional to the
column standard-deviation parameters $\gamma_j$. By symmetry, we get a similar equation for $\beta_j$,
\begin{equation}
  \label{eq:beta}
  \beta_j=\frac{\sum_{i\in\Omega^j}\frac1{\tau_i}(X_{ij}-\alpha_i)}{\sum_{i\in\Omega^j}\frac1{\tau_i}},\quad j=1,\ldots,n,
\end{equation}
where $\Omega^j=\{i| (i,j)\in \Omega\}$, and $m_j=|\Omega^j|\leq m$.

Similarly, the variance conditions for the rows are
\begin{eqnarray}
  \label{eq:tau}
  \frac1{n_i}\sum_{j\in \Omega_i} \widetilde X_{ij}^2&=&  \frac1{n_i}\sum_{j\in \Omega_i}\frac{(X_{ij}-\alpha_i-\beta_j)^2}{\tau_i^2\gamma_j^2}\\ 
&=&1,\nonumber
\end{eqnarray}
which simply says
\begin{equation}
  \label{eq:tau2}
  \tau_i^2= \frac1{n_i}\sum_{j\in \Omega_i}\frac{(X_{ij}-\alpha_i-\beta_j)^2}{\gamma_j^2},\quad i=1,\ldots,m.
\end{equation}
Likewise 
\begin{equation}
  \label{eq:gamma}
    \gamma_j^2= \frac1{m_j}\sum_{i\in \Omega^j}\frac{(X_{ij}-\alpha_i-\beta_j)^2}{\tau_i^2},\quad j=1,\ldots,n.
\end{equation}
The method-of-moments estimators require iterating these four sets of equations (\ref{eq:alpha}), (\ref{eq:beta}), (\ref{eq:tau2}), (\ref{eq:gamma}) till convergence. We monitor the following functions of the  ``residuals''
\begin{eqnarray}
  \label{eq:obj}
  R &=& 
\sum_{i=1}^m\left[\frac1{n_i}\sum_{j\in \Omega_i}\widetilde X_{ij}\right]^2
+
\sum_{j=1}^n\left[\frac1{m_j}\sum_{i\in \Omega^j}\widetilde X_{ij}\right]^2 \\
&&\quad 
+
\sum_{i=1}^m\log^2\left(\frac1{n_i}\sum_{j\in \Omega_i}\widetilde X_{ij}^2\right)
+
\sum_{j=1}^n\log^2\left(\frac1{m_j}\sum_{i\in \Omega^j}\widetilde X_{ij}^2\right)
\end{eqnarray}
In experiments it appears that $R$ converges to zero very fast, perhaps linear convergence.
In Appendix~\ref{sec:moma} we show slightly different versions of these estimators which are more suitable for sparse-matrix calculations.

In practice we may not wish to apply all four standardizations, but instead a subset. For example, we may wish to only standardize columns to have mean zero and variance one.
In this case we simply set the omitted centering parameters to zero, and scaling parameters to one, and skip their steps in the iterative algorithm.
In certain cases we have convergence guarantees:
\begin{itemize}
\item Column-only  centering and/or scaling. Here no iteration is required; the centering step precedes the scaling step, and we are done. Likewise for row-only.
\item Centering only, no scaling. Here the situation is exactly that of an unbalanced two-way ANOVA, and our algorithm is exactly the Gauss-Seidel algorithm for fitting the two-way ANOVA model. This is known to converge, modulo certain degenerate situations.
\end{itemize}
For the other cases we have no guarantees of convergence.

We present an alternative sequence of formulas in Appendix~\ref{sec:moma} which allows one to simultaneously apply the transformations, and learn the parameters.

\section{Discussion}
\label{discuss}
We have presented a new algorithm for matrix completion, suitable for solving (\ref{eq:1}) for very large problems, as long as the solution rank is manageably low. Our algorithm capitalizes on the a different weakness in each of the popular alternatives:
\begin{itemize}
\item \texttt{ALS} has to solve a different regression problem for every row/column, because of their different amount of missingness, and this can be costly. \texttt{softImpute-ALS} solves a single regression problem once and simultaneously for all the rows/columns, because it operates on a filled-in matrix which is complete. Although these steps are typically not as strong as those of \texttt{ALS}, the speed advantage more than compensates.
\item \texttt{softImpute} wastes time in early iterations computing a low-rank SVD of a far-from-optimal estimate, in order to make its next imputation. One can think of \texttt{softImpute-ALS} as simultaneously filling in the matrix at each alternating step,  as it is computing the SVD. By the time it is done, it has the the solution sought by \texttt{softImpute}, but with far fewer iterations.
\end{itemize}
\texttt{softImpute} allows for an extremely efficient distributed implementation (section~\ref{sec:distr-impl}), and hence can scale to large problems, given a sufficiently large computing infrastructure.

\subsubsection*{Acknowledgements}
The authors thank Balasubramanian Narasimhan for helpful discussions on distributed computing in R. The first author thanks Andreas Buja and Stephen Boyd for stimulating ``footnote'' discussions that led to the centering/scaling in Section~\ref{centerscale}.
Trevor Hastie was partially supported by grant  DMS-1407548 from the National
Science Foundation, and grant RO1-EB001988-15 from the National Institutes of
Health.

\bibliographystyle{plainnat}
\bibliography{tibs,trevor,fast_als}

\appendix

\section{Proofs from Section~\ref{conv-rates-ch}}\label{sec:pf-append1}

\subsubsection{Proof of Lemma~\ref{lem:diff-fast-als1}}
To prove this we begin with the following elementary result concerning a ridge regression problem:
\begin{lem}\label{lem:diff-ridge1}
Consider a ridge regression problem 
\begin{equation}\label{ridge-reg-eqn}
H(\beta) : = \half \| y - M \beta \|_2^2 + \frac{\lambda}{2} \| \beta \|_2^2
\end{equation}
with $\beta^*\in \argmin_{\beta} \; H(\beta)$.
Then the following inequality is true:
$$ H(\beta)  - H(\beta^*) = \frac{1}{2} (\beta - \beta^*)^T(M^TM + \lambda \M{I} ) (\beta - \beta^*) = \half \| M(\beta - \beta^*)\|_2^2 + \frac{\lambda}{2} \| \beta - \beta^*\|_2^2$$
\begin{proof}
The proof follows from the second order Taylor Series expansion of $H(\beta)$:
$$H(\beta) = H(\beta^*) + \langle \nabla H(\beta^*), \beta - \beta^* \rangle  +  \frac{1}{2} (\beta - \beta^*)^T(M^TM + \lambda \M{I} ) (\beta - \beta^*)$$
and observing that $ \nabla H(\beta^*)  = 0$.
\end{proof}
\end{lem}

We will need to obtain a lower bound on the difference
$F(A_{k+1}, B_{k}) - F(A_{k}, B_{k})$. Towards this end we make note of the following chain of inequalities:
\begin{align}
F(A_{k}, B_{k}) &=  g(A_{k}B_{k}^T)  +  \frac{\lambda}{2} (  \| A_{k} \|_F^2 +  \| B_{k} \|_F^2 )& \label{line-0}  \\
&= Q_{A} (A_{k} | A_{k} , B_{k} ) &  \label{line-1}\\
&\geq \inf_{Z_{1}} \;\;  Q_A(Z_{1} | A_{k} , B_{k} )&  \label{line-2}  \\
&= Q_A (A_{k+1} | A_{k} , B_{k} ) & \label{line-3} \\
&\geq  g(A_{k+1}B_{k}^T)  + \frac{\lambda}{2} (  \| A_{k+1} \|_F^2 +  \| B_{k} \|_F^2 ) & \label{line-4}\\
&=F(A_{k+1 }, B_{k})& \label{line-5}
\end{align}
where, Line~\eqref{line-1} follows from~\eqref{major-1-2-eq}, and~\eqref{line-4} follows from~\eqref{major-1-2}.

Clearly, from Lines~\eqref{line-5} and~\eqref{line-0} we have~\eqref{line-1-0}
\begin{align}
F(A_{k}, B_{k}) - F(A_{k+1 }, B_{k}) &\geq Q_{A} (A_{k} | A_{k} , B_{k} ) - Q_A (A_{k+1} | A_{k} , B_{k} )& \label{line-1-0} \\
&= \half \| (A_{k+1} - A_{k})B_{k}^T\|_2^2 + \frac{\lambda}{2} \| A_{k+1} - A_{k}\|_2^2, & \label{line-1-1}
\end{align}
where, \eqref{line-1-1} follows from~\eqref{line-1-0} using Lemma~\ref{lem:diff-ridge1}.

Similarly, following the above steps for the $B$-update we have:
\begin{equation}
F(A_{k}, B_{k}) - F(A_{k+1 }, B_{k+1}) \geq \half \| A_{k+1} (B_{k+1} - B_{k})^T\|_2^2 + \frac{\lambda}{2} \| B_{k+1} - B_{k}\|_2^2 \label{line-2-1}.
\end{equation}
Adding~\eqref{line-1-1} and~\eqref{line-2-1} we get~\eqref{upd-als-two-1}
concluding the proof of the lemma.

\subsubsection{Proof of Lemma~\ref{lem-1-stat-1}} \label{proof-lem-1-stat-1}

Let us use the shorthand $\Delta$ in place of 
$\Delta \left (  \left( A, B  \right ) , \left (  A^{+}, B^{+}    \right)    \right )$ as defined in~\eqref{defn:metric1-0}.

First of all observe that the result~\eqref{upd-als-two-1} can be easily replaced with 
$(A_{k}, B_{k}) \leftarrow (A, B)$ and $(A_{k+1}, B_{k+1}) \leftarrow (A^+, B^+)$.
This leads to the following:
\begin{equation}
\begin{aligned}
F(A, B) - F(A^+, B^+) & \geq  \half \left( \| ( A   - A^+) B^T \|_F^2 + \| A^+ (B^+   -  B )^T \|_F^2 \right) \\
&+ \frac{\lambda}{2} \left ( \| A   - A^+  \|_F^2  + \| B^+   -  B \|_F^2 \right).
\end{aligned}
\end{equation}

First of all, it is clear that if $A,B$ is a fixed point then $\Delta=0$.

Let us consider the converse, i.e.,  the case when $\Delta =  0$.
Note that if $\Delta = 0$ then each of the summands appearing in the definition of $\Delta$ is also zero.
We will now make use of the interesting result (that follows from the Proof of Lemma~\ref{lem:diff-fast-als1}) in~\eqref{line-1-0} and~\eqref{line-1-1} which says:
$$ Q_{A} (A | A , B ) - Q_A (A^+ | A , B ) = \half \| (A^+ - A)B^T\|_2^2 + \frac{\lambda}{2} \| A^+ - A \|_2^2 .$$
Now the right hand side of the above equation is zero (since $\Delta = 0$) which implies that,
$Q_{A} (A | A , B ) - Q_A (A^+ | A , B ) = 0$. An analogous result holds true for $B$.

Using the nesting property~\eqref{nesting-1}, it follows that 
$F(A,B) = F(A_+,B_+)$---thereby showing that $(A,B)$ is a fixed point of the algorithm.

\subsubsection{Proof of Theorem~\ref{thm:conv-rate-1}} \label{proof-thm:conv-rate-1}
We make use of~\eqref{upd-als-two-1} and add both sides of the inequality over $k = 1, \ldots, K$, which leads to:
\begin{equation}\label{bound-1-als1}
\sum_{i=1}^{K} \left ( F( A_{k}, B_{k} ) - F( A_{k+1}, B_{k+1} )   \right ) \geq \sum_{k=1}^{K} \eta_{k} \geq K ( \min_{K \geq k \geq 1} \eta_{k})
\end{equation}
Since,  $F( A_k, B_k ) $ is a decreasing sequence (bounded below) it converges to $f^{\infty}$ say. 
\begin{equation}\label{upper-tele-sum1}
\begin{aligned}
\sum_{i=1}^{K} \left ( F( A_k, B_k ) - F( A_{k+1}, B_{k+1} )   \right )   &=& F( A^{1}, B^{1} ) - F( A^{K+1}, B^{K+1} ) \\ 
&\leq& F( A^{1}, B^{1} ) - f^\infty
\end{aligned}
\end{equation}
Using~\eqref{upper-tele-sum1} along with~\eqref{bound-1-als1} we have the following convergence rate:
$$ \min_{1 \leq k \leq K } \eta_{k}  \leq \left ( F( A^{1}, B^{1} ) - F( A^{\infty}, B^{\infty} ) \right) / K, $$
thereby completing the proof of the theorem.

\subsubsection{Proof of Corollary~\ref{cor-1-conv-rate}}\label{proof:cor-1-conv-rate}

Recall the definition of $\eta_{k}$
$$ \eta_{k} =  \half \left( \| (A_{k}   - A_{k+1} ) B_{k}^T \|_F^2 + \| A_{k+1}( B_{k}   -  B_{k+1} )^T \|_F^2 \right) +  \frac{\lambda}{2} \left ( \|A_{k}  - A_{k+1}\|_F^2  + \|  B_{k}   - B_{k+1}\|_F^2 \right)$$
Since we have assumed that
$$ \ell^{U} \M{I} \succeq  B_{k}^T B_{k}  \succeq \ell^{L}\M{I}, \;\;  \ell^{U} \M{I} \succeq  A_{k}^T A_{k}   \succeq \ell^{L}\M{I}, \forall k$$ 
then we have:
$$\eta_{k}  \geq  (\frac{\ell^{L}}{2} + \frac{\lambda}{2})  \|  (A_{k}   - A_{k+1} )\|_F^2 +   (\frac{\ell^{L}}{2} + \frac{\lambda}{2}) 
 \| B_{k}   -  B_{k+1}  \|_F^2$$
Using the above in~\eqref{bound-1-als1} and assuming that $\ell_{L} >0$, we have the bound:
\begin{equation}\label{bound-diff-norm-1-dummy}
 \min_{1 \leq k \leq K } \left( \|  (A_{k}   - A_{k+1} )\|_F^2  +   \| B_{k}   -  B_{k+1}  \|_F^2 \right) \leq   
\frac{2}{(\ell^{L}+ \lambda )}\left(\frac{F( A^{1}, B^{1} ) - f^{\infty}}{K} \right) 
\end{equation}
Suppose instead of the proximity measure: 
$$\left( \|  (A_{k}   - A_{k+1} )\|_F^2  +   \| B_{k}   -  B_{k+1}  \|_F^2 \right),$$ we use the 
proximity measure:
$$\left( \|  (A_{k}   - A_{k+1} )B_{k}^T \|_F^2  +   \| A_{k+1} ( B_{k}   -  B_{k+1})  \|_F^2 \right).$$
Then observing that:
$$  \ell^{U} \|  (A_{k}   - A_{k+1} )\|_F^2 \geq \|  (A_{k}   - A_{k+1} )B_{k}^T\|_F^2, \;\;\; 
 \ell^{U} \| B_{k}   -  B_{k+1}  \|_F^2   \geq \| A_{k+1} ( B_{k}   -  B_{k+1})^T  \|_F^2 $$
we have:
$$ \eta_{k} \geq   \left( \frac{\lambda}{2\ell^{U}} + \frac12\right)\left ( 
\|  (A_{k}   - A_{k+1} )B_{k}^T\|_F^2 + \| A_{k+1} ( B_{k}   -  B_{k+1})  \|_F^2 \right).$$
Using the above bound in~\eqref{bound-1-als1} we arrive at a bound which is similar in spirit to~\eqref{bound-diff-norm-1} 
but with a different proximity measure on the step-sizes:
\begin{equation}\label{bound-diff-norm-2-dummy}
 \min_{1 \leq k \leq K } \left ( 
\|  (A_{k}   - A_{k+1} )B_{k}^T\|_F^2 + \| A_{k+1} ( B_{k}   -  B_{k+1})  \|_F^2 \right) \leq 
\frac{2 \ell^{U}}{\lambda + \ell_U}  \left(\frac{F( A^{1}, B^{1} ) - f^{\infty}}{K} \right)
\end{equation}
It is useful to contrast results~\eqref{bound-diff-norm-1} and~\eqref{bound-diff-norm-2} with the case $\lambda =0$.
\begin{equation}
 \min_{1 \leq k \leq K } \left ( 
\|  (A_{k}   - A_{k+1} )B_{k}^T\|_F^2 + \| A_{k+1} ( B_{k}   -  B_{k+1})  \|_F^2 \right) \leq \begin{cases} 
\frac{2 \ell^{U}}{\lambda + \ell_U}  \left(\frac{F( A^{1}, B^{1} ) - f^{\infty}}{K} \right)  &\text{$\lambda >0$} \\
 2\ell^{U}   \left(\frac{F( A^{1}, B^{1} ) - f^{\infty}}{K} \right) & \lambda = 0 
\end{cases}
\end{equation}
The convergence rate with the other proximity measure on the step-sizes have the following two cases:
\begin{equation}
\begin{aligned}
\min_{1 \leq k \leq K } \left( \|  (A_{k}   - A_{k+1} )\|_F^2  +   \| B_{k}   -  B_{k+1}  \|_F^2 \right) \leq   
\begin{cases}
\frac{2}{(\ell^{L}+ \lambda )}\left(\frac{F( A^{1}, B^{1} ) - f^{\infty}}{K} \right) & \text{$\lambda>0$,} \\
\frac{2}{\ell^{L}}\left(\frac{F( A^{1}, B^{1} ) - f^{\infty}}{K} \right) & \text{$\lambda=0$.} 
\end{cases}
\end{aligned}
\end{equation}

The assumption~\eqref{low-up-bound1}  $\ell^{U} \M{I} \succeq  B_{k}^T B_{k}$ and $\ell^{U} \M{I} \succeq  A_{k}^T A_{k}$
can be interpreted as an upper bounds to the locally Lipschitz constants of the gradients of $Q_{A}(Z | A_{k}, B_{k})$ and 
$Q_{B}(Z | A_{k+1}, B_{k})$ for all $k$:
\begin{equation} \label{eqn-1-1}
\begin{aligned}
\| \nabla Q_{A}(A_{k+1} | A_{k}, B_{k}) - \nabla Q_{A}(A_{k} | A_{k}, B_{k})\| &\leq& \ell^{U}\| A_{k+1} - A_{k}\|, \\
 \| \nabla Q_{B}(B_{k} | A_{k+1}, B_{k}) - \nabla Q_{B}(B_{k+1} | A_{k+1}, B_{k})\| &\leq& \ell^{U} \| B_{k+1} - B_{k}\|. 
 \end{aligned}
 \end{equation}
The above leads to convergence rate bounds on the (partial) gradients of the function $F(A,B)$, i.e., 
$$  \min_{1 \leq k \leq K }  \left(  \| \nabla_{A} f (A_{k},B_{k}) \|^2 +  \| \nabla_{B} f (A_{k+1},B_{k}) \|^2  \right) \leq  
\frac{2(\ell^{U})^2}{(\ell^{L}+ \lambda )}\left(\frac{F( A^{1}, B^{1} ) - f^{\infty}}{K} \right)$$

\subsubsection{Proof of Theorem~\ref{thm-glob-cov1}} \label{proof:thm-glob-cov1}
\begin{proof}
{\noindent Part (a):}\\
We make use of the convergence rate derived in Theorem~\ref{thm:conv-rate-1}. As $k \rightarrow \infty$, it follows that $\eta_{k} \rightarrow 0$.
This describes the fate of the objective values $F(A_k,B_k)$, but does not inform us about the properties of the sequence 
$A_k,B_k$. Towards this end, note that if $\lambda >0$, then the sequence $A_k,B_k$ is bounded and thus has a limit point. 

Let $A_*,B_*$ be any limit point of the sequence $A_k,B_k$, it follows by a simple subsequence argument that 
$F(A_k,B_k) \rightarrow F(A_*,B_*)$ and $A_*,B_*$ is a fixed point of Algorithm \ref{alg:fast-als} and in particular a first order stationary point of problem~\eqref{eq:mmmf}.

{\noindent Part (b):}\\
The sequence $(A_{k}  , B_{k})$ need not have a unique limit point, however for every limit point of 
$B_{k}$ the corresponding limit point of $A_{k}$ must be the same. 

Suppose,  $B_{k} \rightarrow B_*$ (along a subsequence $k \in \nu$). 
We will show that the sequence $A_{k}$ for $k \in \nu $ has a unique limit point.

Suppose there are two limit points of 
$A_{k}$, namely,  $A_1$ and $A_2$ and 
$A_{k_1} \rightarrow A_{1}, k_1 \in \nu_{1} \subset \nu$ and $A_{k_2} \rightarrow A_{2}, k_2 \in \nu_{2} \subset \nu$ with $A_{1} \neq A_{2}$.

Consider the objective value sequence:
$F(A_{k},B_{k})$. For fixed $B_{k}$ the update in $A$ from $A_{k}$ to 
$A^{(k+1)}$ results in 
$$F(A_{k},B_{k}) - F(A_{k+1},B_{k}) \geq \frac{\lambda}{2} \| A_{k} - A_{k+1}\|_F^2$$
Take $k_{1} \in \nu_{1}$ and $k_{2} \in \nu_{2}$, we have:
\begin{align}
F(A_{k_2},B_{k_2}) - F(A_{k_1+1},B_{k_1}) &= \left(  F(A_{k_2},B_{k_2}) -   F(A_{k_2},B_{k_1})   \right) & \nonumber \\
&+   \left( F(A_{k_2},B_{k_1}) - F(A_{k_1+1},B_{k_1}) \right)&  \label{line-line-1}\\
&\geq \left(  F(A_{k_2},B_{k_2}) -   F(A_{k_2},B_{k_1})   \right)  + 
\frac{\lambda}{2} \| A_{k_2} - A_{k_1+1}\|_F^2 & \label{line-line-2}
\end{align}
where Line~\ref{line-line-2} follows by using Lemma~\ref{lem:diff-ridge1}.
As $k_{1}, k_{2} \rightarrow \infty$, 
$B_{k_2},B_{k_1} \rightarrow B_*$ hence,  
$$F(A_{k_2},B_{k_2}) -   F(A_{k_2},B_{k_1}) \rightarrow 0, \;\; \text{and}\;\;   \| A_{k_2} - A_{k_1+1}\|_F^2 \rightarrow \|A_{2} - A_{1}\|_F^2$$
However, the lhs of~\eqref{line-line-1} converges to zero, which is a contradiction. This implies that 
$\|A_{2} - A_{1}\|_F^2 = 0$ i.e. $A_{k}$ for $k \in \nu$  has a unique limit point.

Exactly the same argument holds true for the sequence $A_{k}$, leading to the conclusion of the other part of Part (b).
\end{proof}

\section{Alternative Computing Formulas for Method of Moments}
\label{sec:moma}
In this section we present the same algorithm, but use a slightly different representation.
For matrix-completion problems, this does not make much of a difference in terms of computational load. But we also have other applications in mind, where the large matrix $X$ may be fully observed, but is very sparse. In this case we do not want to actually apply the centering operations; instead we represent the matrix as a ``sparse-plus-low-rank'' object, a class for which we have methods for simple row and column operations.

Consider the row-means (for each row $i$). We can introduce a change $\Delta_i^\alpha$ from the old $\alpha_i^o$ to the new $\alpha_i$.
Then we have
\begin{eqnarray}
  \label{eq:rowmeana}
  \sum_{j\in \Omega_i}\widetilde X_{ij}&=&\sum_{j\in \Omega_i}\frac{X_{ij}-\alpha^o_i-\Delta_i^\alpha-\beta_j}{\tau_i\gamma_j}\\
&=&0\nonumber,
\end{eqnarray}
where as before $\Omega_i=\{j| (i,j)\in \Omega\}$.
Rearranging we get
\begin{equation}
  \label{eq:alphaa}
  \Delta^\alpha_i=\frac{\sum_{j\in\Omega_i}\widetilde {X}_{ij}^o}{\sum_{j\in\Omega_i}\frac1{\tau_i\gamma_j}},\quad i=1,\ldots,m,
\end{equation}
where 
\begin{equation}
  \label{eq:xold}
  \widetilde {X}_{ij}^o=\frac{X_{ij}-\alpha^o_i-\beta_j}{\tau_i\gamma_j}.
\end{equation}
Then $\alpha_i=\alpha_i^o+\Delta_i^\alpha$.
 By symmetry, we get a similar equation for $\Delta^\beta_j$,

Likewise for the variances.
\begin{eqnarray}
  \label{eq:taua}
  \frac1{n_i}\sum_{j\in \Omega_i} \widetilde X_{ij}^2&=&  \frac1{n_i}\sum_{j\in \Omega_i}\frac{(X_{ij}-\alpha_i-\beta_j)^2}{(\tau_i\Delta_i^\tau)^2\gamma_j^2}\\ 
&=& \frac1{n_i}\sum_{j\in \Omega_i}\left(\frac{\widetilde X_{ij}^o}{\Delta^\tau_i}\right)^2\\
&=&1.\nonumber
\end{eqnarray}
Here we modify $\tau_i$ by a multiplicative factor $\Delta^\tau_i$.
Here the solution is
\begin{equation}
  \label{eq:tau2a}
  (\Delta^\tau_i)^2= \frac1{n_i}\sum_{j\in \Omega_i}(\widetilde {X}^o_{ij})^2,\quad i=1,\ldots,m.
\end{equation}
 By symmetry, we get a similar equation for $\Delta^\gamma_j$,

The method-of-moments estimators amount to iterating these four sets of equations till convergence. Now we can monitor the changes via
\begin{eqnarray}
  \label{eq:obj}
  R &=& 
\sum_{i=1}^m {\Delta^\alpha_i}^2
+
\sum_{j=1}^n {\Delta^\beta_j}^2
+
\sum_{i=1}^m \log^2\Delta^\tau_i
+
\sum_{j=1}^n \log^2\Delta^\gamma_j
\end{eqnarray}
which should converge to zero.

\end{document}